\documentclass[peerreview]{IEEEtran}
\IEEEoverridecommandlockouts
\usepackage{hyperref} %
\usepackage%
{hyperref} %
\hypersetup{
    colorlinks,
    linkcolor={blue!80!black},%
    citecolor={green!30!black},
    urlcolor={blue!80!black}
}

\usepackage{amsmath,amssymb,amsfonts}
\usepackage{amsmath}
\usepackage{bbold}
\usepackage{nicefrac}
\usepackage{amsthm}
\usepackage{mathrsfs} 

\usepackage{algorithmic}

\usepackage{textcomp}
\def\BibTeX{{\rm B\kern-.05em{\sc i\kern-.025em b}\kern-.08em
    T\kern-.1667em\lower.7ex\hbox{E}\kern-.125emX}}

\usepackage{physics}

\usepackage{graphicx}
\usepackage{xcolor}

\usepackage{comment}

\renewcommand{\trace}{\mathrm{Tr}}

\theoremstyle{remark}	\newtheorem{theorem}{Theorem}
\theoremstyle{remark}	\newtheorem{lemma}[theorem]{Lemma}
\theoremstyle{remark}	
\theoremstyle{remark}	
\theoremstyle{remark} \newtheorem{definition}{Definition}
\theoremstyle{remark} \newtheorem{remark}{Remark}
\theoremstyle{remark} 

\usepackage[square,numbers,sort&compress]{natbib}
\allowdisplaybreaks %

\begin{document}

\title{Entanglement-Assisted Covert Communication via Qubit Depolarizing Channels%
}

\author{
		\vspace{0.1cm}
    \IEEEauthorblockN{
    Elyakim Zlotnick\IEEEauthorrefmark{1}, Boulat Bash\IEEEauthorrefmark{2}, and Uzi Pereg\IEEEauthorrefmark{1}} \\
		\vspace{0.25cm}
    \IEEEauthorblockA{\normalsize \IEEEauthorrefmark{1}Electrical and Computer Engineering and Hellen-Diller Quantum Center, Technion -- Israel Institute of Technology \\
    \IEEEauthorrefmark{2}Electrical and Computer Engineering, University of Arizona%
    }
    }
\maketitle
\begin{abstract}
We consider entanglement-assisted communication over the qubit depolarizing channel under the security requirement of covert communication, where 
the transmission itself
must be concealed from detection by an adversary.
Previous work showed that $O(\sqrt{n})$ information bits can be reliably and covertly transmitted in $n$ channel uses without entanglement assistance.  However, Gagatsos et al. (2020) showed that entanglement assistance can increase this scaling to $O(\sqrt{n}\log{n})$ for continuous-variable bosonic channels. Here, we present a finite-dimensional parallel, and show that $O(\sqrt{n}\log{n})$ covert bits can be transmitted reliably over $n$ uses of a qubit depolarizing channel.
The coding scheme employs ``weakly" entangled states 
such that their squared amplitude scales as $O\left(\nicefrac{1}{\sqrt{n}}\right)$. 
\end{abstract}


\section{Introduction}
Privacy and confidentiality are critical in communication systems \cite{wang2020security}.
Traditional security approaches (e.g., encryption 
\cite{talbot2006complexity}, %
information-theoretic secrecy 
\cite{bloch2011physical}, %
and quantum key distribution \cite{bennet1984quantum,renner2008security,scarani2009security}) ensure that an eavesdropper is unable to recover any transmitted information. 
However,
privacy and safety concerns may further require \emph{covertness} \cite{9380147,7217803}. 
Covertness is a stringent requirement whereby 
%
the transmission itself is concealed from detection by an adversary (a warden) \cite{bash2015hiding,8917582}. 
Despite the severity of limitations imposed by covertness,
it is possible to  communicate $O(\sqrt{n})$  bits of information both reliably and covertly over $n$ classical channel uses \cite{Bash_first_classic,Bloch,wang2016fundamental}. %
This property is referred to as the ``square root law” (SRL).
The SRL has also been observed in covert communication over 
finite-dimensional classical-quantum channels \cite{Bash-Quantum, wang16cq-srlconverse, Bash-Quantum-revision,tahmasbi2019framework}, as well as continuous-variable  bosonic channels \cite{bash_first_bosonic,8976410,Bash-Bosonic,9834394}. %
Covert sensing is also governed by an SRL \cite{bash2017fundamental_sensing, goeckel17sensinglinsystems-asilomar, tahmasbi2021covert}.
Other covert models are studied in
\cite{Deng2022,ZivariFard2022,Yang2022,Amihood:ITW2022,hayashi2023covert,bounhar2023mixing}.

Proving the achievability of the SRLs discovered so far involves the following principles.
In the finite-dimensional case, both classical and quantum 
\cite{Bloch,wang2016fundamental,Bash-Quantum,Bash-Quantum-revision,tahmasbi2019framework},  a symbol (say, $0$) in the input alphabet is designated as ``innocent.'' The codebook is generated such that %
a non-innocent symbol is transmitted with probability $\sim\nicefrac{1}{\sqrt{n}}$ to ensure covertness.
On the other hand, the innocent symbol corresponding to zero transmitted power occurs naturally in the continuous-variable covert communication over classical additive white Gaussian noise (AWGN) \cite{Bash_first_classic,Bloch,wang2016fundamental} and classical-quantum bosonic \cite{bash_first_bosonic,8976410,Bash-Bosonic,9834394} channels.
Maintaining average transmitted power $O(\nicefrac{1}{\sqrt{n}})$ correspondingly measured in Watts and in the emitted photon number ensures covertness.

Pre-shared entanglement resources are known to increase performance and throughput %
\cite{bennett1999entanglement,EA_capacity_bennet,hao2021entanglement,chiuri2013experimental,9319007}. 
Gagatsos et al. 
\cite{Bash-Bosonic} showed that entanglement assistance allows transmission of $O(\sqrt{n}\log{n})$ reliable and covert bits over $n$ uses of continuous-variable bosonic channel, surpassing the SRL scaling (see also \cite{bloch_resource_efficient}).
%
As in the unassisted setting, the transmission is limited 
to $O(\nicefrac{1}{\sqrt{n}})$ mean photon number.
However, so far it has remained open whether such a performance boost can be achieved in communication over finite-dimensional quantum channels.

The depolarizing channel is a fundamental model  that has gained significant attention in both experimental \cite{ChiuriGiacominiMacchiavelloMataloni:13p,Google:19p} and theoretical \cite{king2003capacity,
depolarizing_comp}  research. Depolarization may be regarded as the worst type of noise  in a quantum system and can also be interpreted as the result of a random unitary error with a probability law that follows the Haar measure or, alternatively, a random Pauli error.
Furthermore,
the insights on the depolarizing channel are often useful in the derivation of results for a general quantum channel \cite{bennett1999entanglement}, \cite[Sec. 11.9.1]{wilde2013quantum}.

Here, we show that entanglement assistance enables reliable and covert transmission of $O(\sqrt{n}\log{n})$ bits in $n$ uses of a finite-dimensional qubit depolarizing channel. The entanglement-assisted covert communication scheme is illustrated in Figure~\ref{Figure:Covert_EA}.
Our analysis is fundamentally different from the previous works. %
In particular, we do \emph{not} encode a random bit sequence with $\sim\nicefrac{1}{\sqrt{n}}$ frequency (or probability) of non-innocent symbols. Instead, we employ ``weakly" entangled states of the form
\begin{align}
\ket{\psi_{A_1 A}}=
\sqrt{1-\alpha}\ket{00}+\sqrt{\alpha}\ket{11},
\label{eq:Superposition}
\end{align}
such that the squared amplitude of this  quantum superposition of states describing innocent and non-innocent symbols is $\alpha=O\left(\nicefrac{1}{\sqrt{n}}\right)$. 
The labels $A_1$ and $A$ correspond to a reference system and to the channel input system, respectively. The former can be interpreted as Bob's share of the entanglement resource.
 The idea is inspired by a recent work on non-covert communication  showing that 
 controlling $\alpha\in [0,1]$ 
 using states in \eqref{eq:Superposition} can outperform time division \cite{9834764}.
 To show covertness, we observe that tracing out the resource system $A_1$ from $\ket{\psi_{A_1 A}}$ results in a state identical to the one in unassisted scenario from \cite{Bash-Quantum, Bash-Quantum-revision}.

 The paper is organized as follows. 
 In Section II, the definitions and channel model are provided, including  notation, an overview of the system and coding, and a presentation of the covert communication problem.
 The results are described in Section III, with the main achievability proof in Section IV and technical details deferred to the appendices. 
Section V presents interpretation through energy-constrained communication, and Section V concludes with a summary and discussion.

\begin{figure*}[tb]
\includegraphics[scale=0.84,trim={2.5cm 10.5cm 0 10cm},clip]{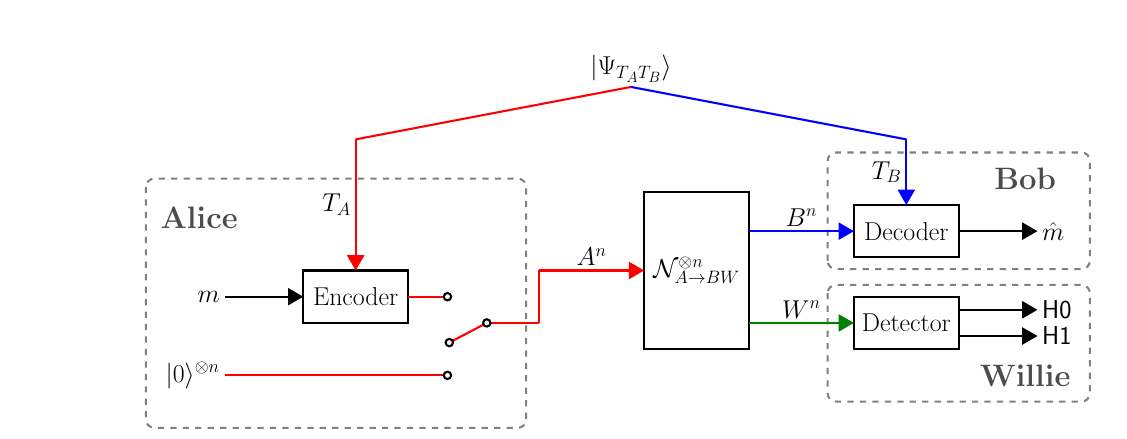} %
\caption{Entanglement-assisted coding for covert communication over a quantum  channel $\mathcal{N}_{A\rightarrow B W}$.
%
%
Alice and Bob  access entangled resources in systems $T_{A}$ and $T_{B}$, respectively.
Message $m$ is encoded by applying the map $\mathcal{F}^{(m)}_{  T_A \rightarrow A^n}$ to the entangled system $T_A$. 
Alice decides whether to transmit to Bob (Case 1) or not (Case 0). A switch connects the channel to the encoder in Case 1 or to a zero sequence 
$\ket{0}^{\otimes n}$ in Case 0.
Alice transmits the systems $A^n$ over the quantum channel.
Bob receives the channel output systems $B^n$,  and performs a joint decoding measurement on the systems $B^n$ and $T_B$, using a POVM $\mathcal{D}_{B^n T_B}  $. Willie receives the output systems $W^n$, and performs a binary measurement to test whether transmission has taken place.
}
\label{Figure:Covert_EA}
\end{figure*}

\section{Definitions and Channel Model}
\subsection{Notation}
\label{sec:notation}
We use standard notation  in quantum information processing,
as, e.g., in \cite[Ch.~2.2.1]{tomamichel2015quantum}. The Hilbert space for system $A$ is denoted by $\mathcal{H}_A$. The space of linear operators (resp. density operators) $\mathcal{H}\to \mathcal{H}$ is denoted by $\mathcal{L}(\mathcal{H})$ (resp. $\mathscr{S}(\mathcal{H})$). %
A positive operator-valued measure (POVM) $\{ D_m\}_{m=1}^M$ is a set of positive semidefinite linear operators in $\mathcal{L}(\mathcal{H})$ such that $\sum_{m=1}^M D_m = \mathbb{1}$, where
$\mathbb{1}$ is the identity operator on $\mathcal{H}$.

Given a pair of 
quantum states $\rho,\sigma\in \mathscr{S}(\mathcal{H})$, the quantum relative entropy is defined as  $D(\rho||\sigma) = \trace[\rho(\log(\rho)-\log(\sigma )]$, if $  \mathrm{supp}(\rho)\subseteq \mathrm{supp}(\sigma)$; and $D(\rho||\sigma)=+\infty$, otherwise.
In addition, for a spectral decomposition  $\sigma = \sum_i \lambda_i P_i$,  let \cite{tahmasbi2021covert}:
\begin{align}
    \eta(\rho||\sigma) &= \sum_{i \neq j} \frac{\log(\lambda_i) - \log(\lambda_j)}{\lambda_i-\lambda_j}\Tr [(\rho-\sigma)P_i(\rho-\sigma)P_j] 
    \nonumber\\& 
    \phantom{=}+ \sum_i \frac{1}{\lambda_i}\Tr [(\rho-\sigma)P_i(\rho-\sigma)P_i] \,.\label{eq:eta}
\end{align}
%
Given a bipartite state $\rho_{AB}$, 
 the quantum mutual information is defined as
$
I(A;B)_\rho=H(\rho_A)+H(\rho_B)-H(\rho_{AB}) 
$, where $
H(\rho) \equiv  -\trace[ \rho\log \rho  ]
$ denotes the von Neumann entropy for a density operator $\rho$. 
Furthermore, the conditional quantum entropy 
is defined by
$H(A|B)_{\rho}=H(\rho_{AB})-H(\rho_B)$. 

A quantum  channel is defined  as a  completely-positive  trace-preserving (CPTP) linear map $\mathcal{N}_{A\to B}:\mathcal{L}(\mathcal{H}_A)\to \mathcal{L}(\mathcal{H}_B)$.
Every quantum channel has a Stinespring representation,
$\mathcal{N}_{A\to B}(\rho)=\Tr_E(V\rho V^\dag)$, for $\rho\in\mathcal{L}(\mathcal{H}_A)$, 
where the operator $V:\mathcal{H}_A \to \mathcal{H}_B \otimes \mathcal{H}_E$ is an isometry. %

For a given function $g(n)$, we denote by $O\big(g(n)\big)$ the set of functions $f(n)$ for which there exist positive constants $c$ and $n_0$ such that $ \abs{f(n)} \leq c g(n)$ for all $n \geq n_0$. We write $f(n)=O\big(g(n)\big)$ to indicate that a function $f(n)$ belongs to the set $O\big(g(n)\big)$ \cite{cormen2022introduction}.
%
Equivalently, 
 \begin{align}
 f(n)= O\big(g(n)\big) \text{ if }%
\limsup_{n\to\infty} \left|\frac{f(n)}{g(n)}\right|<\infty \,.
\label{Equation:bigOzero}
\end{align}
%
%
%
%
%
Similarly, for continuous-variable functions,
 $F$ and $G$ on
$\in[0,\infty)$, we write
 \begin{align}
 F(x)=\mathcal{O}\big(G(x)\big) \text{ if }%
\limsup_{x\to 0} \left|\frac{F(x)}{G(x)}\right|<\infty \,.
\label{Equation:bigOzeroContinuous}
\end{align}

Additionally, for a given function $g(n)$, we denote by $\omega\big(g(n)\big)$ the set of functions $f(n)$, where for all positive constants $c$, there exists $n_0$ such that $0\leq c g(n) \leq f(n)$ for all $n \geq n_0$. We write $f(n)=\omega\big(g(n)\big)$ to indicate that a function $f(n)$ belongs to the set $\omega\big(g(n)\big)$. Equivalently,
\begin{align}
    f(n)= \omega\big(g(n)\big) \text{ if }%
\lim_{n\to\infty} \frac{f(n)}{g(n)} =\infty \,.
\end{align}
Similarly, for a given function $g(n)$, we denote by $o\big(g(n)\big)$ the set of functions $f(n)$, where for all positive constants $c$, there exists $n_0$ such that $0\leq f(n)\leq c g(n) $ for all $n \geq n_0$. We write $f(n)=o\big(g(n)\big)$ to indicate that a function $f(n)$ belongs to the set $o\big(g(n)\big)$. Equivalently,
\begin{align}
    f(n)= o\big(g(n)\big) \text{ if }%
\lim_{n\to\infty} \frac{f(n)}{g(n)} =0 \,.
\end{align}

\subsection{Channel Model}
\label{Subsection:channel_model}
Consider a covert communication quantum channel
$\mathcal{N}_{A\to BW}$, which maps a quantum input state $\rho_A$  to a joint  output state $\rho_{BW}$.
The systems $A$, $B$, and $W$ are associated with the transmitter, the legitimate receiver, and an adversarial warden, referred to as Alice, Bob, and Willie. %
The marginal channels $\mathcal{N}_{A\to B}$ and $\mathcal{N}_{A\to W}$, from Alice to Bob, and from Alice to Willie, respectively, satisfy
$%
\mathcal{N}_{A\to B}(\rho_A)=\Tr_{W}\left( \mathcal{N}_{A\to BW}(\rho_A)\right) %
$ and $%
\mathcal{N}_{A\to W}(\rho_A)=\Tr_{B}\left( \mathcal{N}_{A\to BW}(\rho_A)\right) %
$ %
for %
$\rho_A\in \mathscr{S}(\mathcal{H}_A)$.
Our channel is  memoryless: for $\rho_{A^n}$ occupying input systems $A^n=(A_1,
\ldots, A_n)$, the joint output state is 
$\mathcal{N}^{\otimes n}_{A\to BW}(\rho_{A^n})$.

The depolarizing channel is a natural  model for noise in quantum systems
\cite{bennett1999entanglement,king2003capacity,depolarizing_comp}.
The qubit depolarizing channel  with parameter $q$
transmits the input qubit perfectly with probability $1 - q$, and outputs a completely mixed
state with probability $q$.
Consider a qubit depolarizing channel from Alice to Bob expressed as:
\begin{align}
&\mathcal{N}_{A\to B}(\rho_A)=
(1-q)\rho_A+q\frac{\mathbb{1}}{2}
\nonumber\\
&= \left( 1-\frac{3q}{4} \right)\rho_A%
\label{eq:twirl}%
+\frac{q}{4}\left(X\rho_A X+Y\rho_A Y+Z\rho_A Z\right),
\end{align}
where $0<q<1$, with dimensions $\dim(\mathcal{H}_A)=\dim(\mathcal{H}_B)=2$,
$X$, $Y$, and $Z$ are the Pauli operators, %
and \eqref{eq:twirl} follows from the Pauli twirl identity \cite[Ch.~4.7.4]{wilde2013quantum}.
Here, we investigate covert communication over a depolarizing channel
$\mathcal{V}_{A\to BE_1 E_2}$ given by
the Stinespring dilation:
\begin{align}
\mathcal{V}_{A\to BE_1 E_2}(\rho_A)=
V\rho_A V^\dag,
\end{align}
where 
$V:\mathcal{H}_A \to \mathcal{H}_B \otimes \mathcal{H}_{E_1} \otimes \mathcal{H}_{E_2}$ is an isometry defined by %
\begin{align}
    V &\equiv \sqrt{1-\frac{3q}{4}}\mathbb{1} \otimes \ket{00} 
    +  \sqrt{\frac{q}{4}}X\otimes \ket{01} 
      \nonumber\\&  
    + \sqrt{\frac{q}{4}}Y\otimes \ket{11} + \sqrt{\frac{q}{4}} Z\otimes \ket{10} \,.
\label{Equation:Depolarizing_Stinespring}
\end{align}
\begin{remark}
The canonical Stinespring dilation  for the qubit depolarizing channel is defined by $\widetilde{\mathcal{V}}_{A\to BE}(\rho)=\widetilde{V}\rho \widetilde{V}^\dag$, where
$\widetilde{V}\equiv \sqrt{1-\frac{3q}{4}}\mathbb{1} \otimes \ket{0} 
    +  \sqrt{\frac{q}{4}}X\otimes \ket{1} 
    + \sqrt{\frac{q}{4}}Y\otimes \ket{2} + \sqrt{\frac{q}{4}} Z\otimes \ket{3}$ (see \cite[Eq. (13)]{depolarizing_comp}).
    For $E\equiv (E_1,E_2)$, our definition in \eqref{Equation:Depolarizing_Stinespring} is equivalent to this canonical description.
    Note, however, that any other Stinespring representation is equivalent to \eqref{Equation:Depolarizing_Stinespring} up to an isometry on the environment $E$ %
    \cite[Sec. III-B]{4475375}.
\end{remark}

We consider three cases:
\begin{itemize}
\item
Scenario 1: Willie receives $(E_1,E_2)$

\item 
Scenario 2: Willie receives $E_2$

\item 
Scenario 3: Willie receives $E_1$

\end{itemize}

\begin{remark}
\label{Remark:Worst_Case}
In any depolarizing channel model,
Scenario 1  represents the worst-case scenario where  
  Willie is given access to Bob's entire environment, 
$E=(E_1,E_2)$.
This allows Willie to acquire maximum information in the quantum setting. 
%
It is important to note that quantum no-cloning theorem prohibits Willie from receiving a copy of Bob's output state, whereas in the classical setting, Willie could have a copy of Bob's output. 
Hence, 
 the quantum channel from Alice to Willie is \emph{not} a depolarizing channel.
\end{remark}

\begin{remark}
\label{Remark:bounds_of_q}
    In the boundary case of $q=0$, Bob receives the qubit state as is, while Willie obtains no information, in agreement with the no-cloning theorem.
    Essentially, there is no warden in this case, hence we may transmit 
    $O(n)$ bits, and achieve a positive Shannon rate in bits per channel use.
    Conversely, if $q=1$, Willie receives the qubit state, and Bob gets only noise, rendering any communication impossible.
\end{remark}

\begin{remark}
Scenarios 2 and 3 can be practically motivated by Willie's instruments not having access to the entirety of Alice and Bob's environment.
While the model specification of Willie's observation may seem artificial, it allows us to demonstrate interesting properties of covert communication with entanglement assistance. 
We argue that covert communication is impossible in Scenario~1, while in Scenario~2, Alice can transmit $O(n)$ covert bits to Bob.
Yet, Scenario 3 is the most interesting case, where entanglement assistance increases the scale of information bits from
$O(\sqrt{n})$ to $O(\sqrt{n} \log n)$.
We observe that the performance does not only depend on the dimension, as Willie receives a single qubit in both Scenarios 2 and 3, yet the behavior is completely different.
Further details are given in the Results section (see Section~\ref{Section:Results}).
\end{remark}

\subsection{Entanglement-assisted Code}
\label{sec:code}
The definition of a code for covert communication over a quantum channel with entanglement assistance is given below.
\begin{definition}
An $(M,n)$ entanglement-assisted code $(\Psi,\mathcal{F},\mathcal{D})$  consists
of: a message set $[1:M]$, where $M$ is an integer, a pure entangled state $\Psi_{T_A T_B}$, a collection of encoding maps $\mathcal{F}_{T_A \to A^n}^{(m)}:\mathscr{S}(\mathcal{H}_{T_A}) \to \mathscr{S}(\mathcal{H}_A^{\otimes n})$ for $m \in [1:M]$, and a decoding POVM $\mathcal{D}_{B^n T_B} = \{D_m \}_{m=1}^M$.
\end{definition}

The communication setting is depicted in Figure~\ref{Figure:Covert_EA}.
Suppose that Alice and Bob share the entangled 
state $\Psi_{T_A T_B}$, in systems $T_A$ and $T_B$, respectively. %
Alice wishes to send one of $M$ equally-likely %
messages.
To encode a message $m$, she applies the encoding map 
$\mathcal{F}_{T_A \to A^n}^{(m)}$ to her share $T_A$ of the entanglement resource. This results in a quantum state
$\rho_{A^n T_B}^{(m)} =  (\mathcal{F}_{T_A^n \to A^n}^{(m)}\otimes \mathbb{1}_{T_B})(\Psi_{T_A T_B})$.

Alice decides whether to transmit to Bob (Case 1), or not (Case 0). The innocent state is $\ket{0}$; any other state is non-innocent. 
She does not transmit in Case 0: the channel input is $\ket{0}^{\otimes n}$.
In Case 1, she transmits part of $\rho_{A^n T_B}^{(m)}$ occupying systems $A^n$ through $n$ uses of the covert communication channel $\mathcal{N}_{A\to BW}$. The joint output state is
$\rho_{B^n W^n T_B}^{(m)} = \left(\mathcal{N}_{A \to BW}^{\otimes n} \otimes \mathrm{id}_{T_B}\right)\left(\rho_{A^n T_B}^{(m)}\right)$. 
Bob decodes the message from the reduced output state $\rho_{B^n T_B}^{(m)}=\Tr_{W^n}\left[\rho_{B^n W^n T_B}^{(m)}\right]$ by applying the POVM 
$\mathcal{D}_{B^n T_B}$.%

\begin{remark}
    We assume without loss of generality that the innocent state is represented by $\ket{0}$. However, it is important to note that this choice is arbitrary. Since the depolarizing channel is symmetric with respect to the input state, our findings can easily be extended to any product state $\ket{\psi_\text{idle}}^{\otimes n}$ that corresponds to an idle transmission system.
\end{remark}
 
\begin{remark}
In our achievability analysis, we  identify the entanglement resource $\Psi_{T_A T_B}$ with the product state 
$
\psi_{A_1 A}^{\otimes n}$, as in \eqref{eq:Superposition}. That is, we use entanglement resources such that 
Alice and Bob's entangled systems, $T_A$ and $T_B$, 
consist of $n$ copies of $A$ and $A_1$, respectively.
\end{remark}

 
%
%
%
%
%
%
%
%
%

\subsection{Reliability and Covertness}
We characterize reliability by 
the average probability of decoding error for entanglement-assisted code $(\Psi,\mathcal{F},\mathcal{D})$ defined in Section \ref{sec:code}:
\begin{align}
P_e^{(n)}(\Psi,\mathcal{F},\mathcal{D})=
\frac{1}{M}\sum_{m=1}^M \trace\left[(\mathbb{1}-D_m)\rho_{B^n T_B}^{(m)}\right]
\end{align}
where $\rho_{B^n T_B}^{(m)}$ is the reduced state
of the joint output state. %

Willie does not have access to Alice and Bob's entanglement resource and receives the reduced output state $\rho_{W^n}^{(m)}=\Tr_{B^n T_B}\left[\rho_{B^n W^n T_B}^{(m)}\right]$ occupying the system $W^n$.
Willie has to determine whether  Alice transmitted to Bob. To this end, he performs a binary measurement $\{\Delta_{\mathsf{H0}},\Delta_{\mathsf{H1}}\}$, where the outcome $\mathsf{H1}$ represents the hypothesis that Alice sent information, while $\mathsf{H0}$ indicates the contrary hypothesis.

He fails by either accusing Alice of transmitting when she is not (false alarm), or missing
Alice’s transmission (missed detection).
Denoting the probabilities of these errors by $P_{\text{FA}} =
P(\text{choose~} \mathsf{H1}|\mathsf{H0} \text{~is true})$ and $P_{\text{MD}} = P(\text{choose~} \mathsf{H0}|\mathsf{H1} \text{~is true})$, respectively, and assuming equally
likely hypotheses, %
Willie’s average probability of error is %
$%
    P_w^{(n)}%
    = \frac{P_{\text{FA}} + P_{\text{MD}}}{2}
$. %
A random choice  %
yields an ineffective detector with $P_w^{(n)}%
= \frac{1}{2}$. 
The
goal of covert communication is to design a sequence of codes such that Willie’s detector is
forced to be arbitrarily close to ineffective. 
Denote the average state that Willie receives by
\begin{align}
\overline{\rho}_{W^n}=\frac{1}{M}\sum_{m=1}^M \rho^{(m)}_{W^n} 
\label{Equation:Willie_Average_State}
\end{align}
where $\rho_{W^n}^{(m)}$ is the reduced state
of the joint output  $\rho_{B^n W^n T_B}^{(m)}$.
A sufficient condition \cite{Bash-Quantum, Bash-Quantum-revision} to render \emph{any} detector ineffective for Willie is
$%
D(\overline{\rho}_{W^n}||\omega_0^{\otimes n}) \approx 0
$, %
where $\omega_0\equiv \mathcal{N}_{A\to W}(\ketbra{0})$ is the output corresponding to innocent input.
Formally, an 
$(M,n,\varepsilon,\delta)$-code for entanglement-assisted covert communication satisfies
\begin{align}
P_e^{(n)}(\Psi,\mathcal{F},\mathcal{D})\leq \varepsilon
\intertext{and}
D(\overline{\rho}_{W^n}||\omega_0^{\otimes n})\leq \delta \,.
\label{Equation:Covertness_Requirement}
\end{align}

\subsection{Capacity}
\label{Subsection:Rate}
In traditional communication problems, the coding rate is defined as 
$R=\frac{\log(M)}{n}$, i.e., the number of bits per channel use. %
In covert communication, however, the best achievable rate is zero, since the number of information bits  is sublinear in $n$. Here we prove that entanglement assistance allows reliable transmission of
$\log(M)=O(\sqrt{n}\log{n})$ covert bits. Hence, 
the covert coding rate is characterized as in \cite{Bash-Bosonic}:
\begin{align}
    L &=  \frac{\log(M)}{\sqrt{\delta n}\log{n}}.
    \label{Equation:covert_rate_definition_delta}
\end{align}
where $\delta$ is the covertness level in \eqref{Equation:Covertness_Requirement}.

\begin{definition}
\label{Definition:Achievable_Rate}
A covert rate $L> 0$ is achievable with entanglement assistance  if for every $\varepsilon,\delta>0$, and sufficiently large $n$, there exists a
$(2^{L\sqrt{\delta n}\log{n}},n,\varepsilon,\delta)$  code. %
\end{definition}

\begin{remark}
Achievable rates  correspond to error and covertness levels that tend to zero in the limit of $n\to\infty$.
That is, one may rewrite 
Definition~\ref{Definition:Achievable_Rate}
as follows \cite{Bash-Quantum}.
A rate $L$ is asymptotically achievable if there exists a sequence of codes such that 
\begin{align}
       \frac{\log(M)}{\log{n}\sqrt{nD(\overline{\rho}_{W^n}||\omega_0^{\otimes n}})}
     &\geq L-\zeta_n
     \quad \forall n\geq n_0\label{Equation:covert_rate_definition}
 \intertext{
 for some $n_0>0$ and sequence $\zeta_n$ that tends to zero as $n\to\infty$,
 while the error probability  satisfies}
\lim_{n\to\infty} P_e^{(n)}(\Psi,\mathcal{F},\mathcal{D})&=0\,,
\intertext{and the covertness,}
\lim_{n\to\infty} D(\overline{\rho}_{W^n}||\omega_0^{\otimes n})&=0 \,.
 \end{align}

\end{remark}

\begin{definition}
\label{Definition:Capacity}
The entanglement-assisted covert capacity is defined as the supremum of achievable covert rates. We denote this capacity by $C_{\text{cov-EA}}(\mathcal{N})$, where the subscript stands for  covert communication with entanglement assistance. 
\end{definition}

Consider the following state, with $\alpha\in [0,1]$: 
\begin{align}
\varphi_{\alpha} &\equiv (1-\alpha) \ketbra{0} + \alpha\ketbra{1} \,.
\label{Equation:Bash_Quantum_State}
\end{align}
Let $\gamma_n=o(1)\cap\omega\left(\frac{(\log n)^{\frac{4}{3}}}{n^{1/6}}\right)$, that is, as $n\to\infty$, $\gamma_n \to 0$ and $\frac{n^{1/6}\gamma_n}{\log n} \to +\infty$.
Choosing $\alpha=\alpha_n$ where
\begin{align}
     \alpha_n&\equiv\frac{\gamma_n}{\sqrt{n}}\
     \label{Equation:alpha_definition}
 \end{align}
 ensures covertness \cite{Bash-Quantum, Bash-Quantum-revision}.
That is, if the average state of the input system $A^n$ is given by
$%
\rho_{A^n}=(\varphi_{\alpha_n})^{\otimes n} %
$, %
then the covertness requirement \eqref{Equation:Covertness_Requirement} is satisfied for  large $n$.

\section{Results}
\label{Section:Results}
We address three scenarios presented in Section~\ref{Subsection:channel_model}. We begin with the case where Willie receives the entire environment, i.e., both $E_1$ and $E_2$. This can be viewed as the worst-case scenario (see Remark~\ref{Remark:Worst_Case}).

\begin{theorem}
\label{Theorem:willie_recievs_all_E}
Covert communication is impossible in \mbox{Scenario 1}. Hence, if $W=(E_1,E_2)$, then $C_{\text{cov-EA}}(\mathcal{N})=0$.
\end{theorem}

\begin{proof}[Proof of Theorem \ref{Theorem:willie_recievs_all_E}]
Let $\omega_0$ and $\omega_1$ denote Willie's output states corresponding to the inputs $\ket{0}$ and $\ket{1}$, respectively. That is $\omega_x \equiv \mathcal{N}_{A\to W}(\ketbra{x})$ for $x\in\{0,1\}$.

    In this scenario, we have $\mathrm{supp}(\omega_1)\not\subseteq \mathrm{supp}(\omega_0)$. We show this in detail in Appendix \ref{Subsection:willie_recives_all_E}
    . Therefore,  Willie can perform a measurement to detect a non-zero transmission with certainty.
\end{proof}
    Essentially, in Scenario 1, Willie's  entanglement with the transmitted qubit is strong enough for him to detect any encoding operation. %

Next, we
consider another extreme setting. %

\begin{theorem}
\label{Theorem:willie_recievs_second_qubit}
Covert communication is trivial in Scenario 2.
That is, if
     $W=E_2$, then Alice can communicate unconstrained by the covertness requirement, and transmit $O(n)$ bits. 
\end{theorem}

\begin{proof}[Proof of Theorem \ref{Theorem:willie_recievs_second_qubit}]
If $W=E_2$, then Willie receives  $\omega_0 = \omega_1=
 \left(1-\frac{q}{2}\right)\ketbra{0} + \frac{q}{2} \ketbra{1}
 $ (see Appendix \ref{Subsection:willie_recives_second_qubit}).
    In this scenario, even without entanglement assistance,  Alice can transmit classical codewords as in the standard non-covert model, while  Willie cannot discern between zero and non-zero inputs. %
\end{proof}

\begin{figure}[tb]
\centering %
\includegraphics[scale=0.5,trim={4.5cm 3.5cm 0cm 4.5cm},clip]
{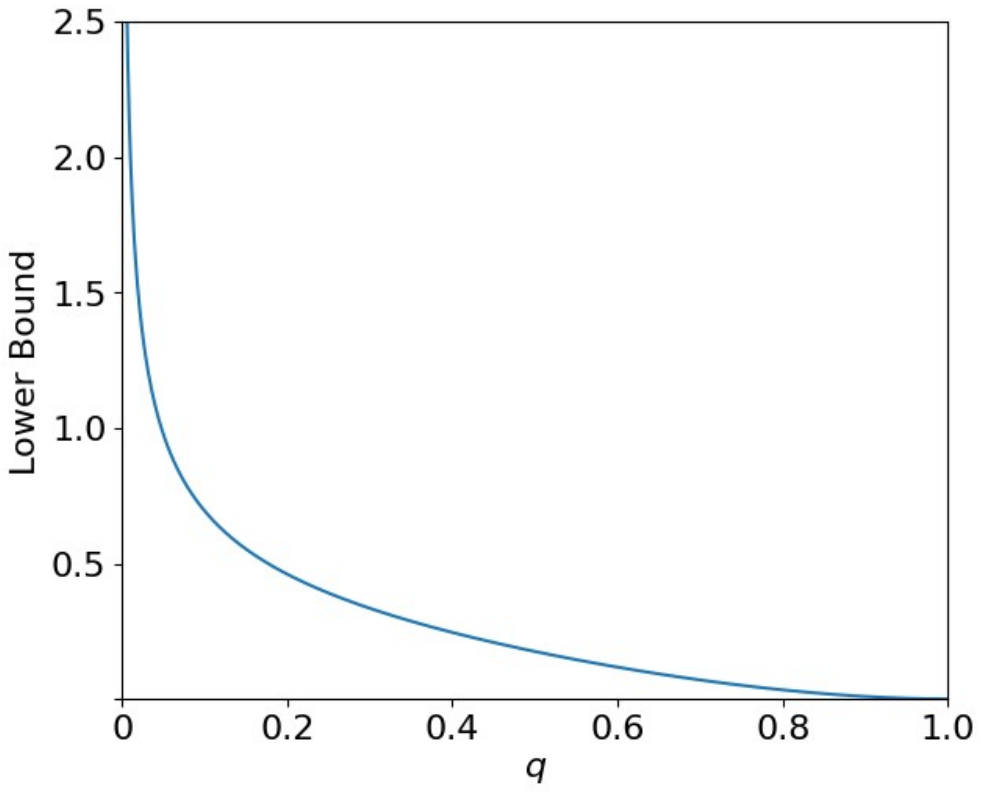} %
\caption{ The lower bound on the entanglement-assisted covert capacity of Scenario 3 in 
Theorem~\ref{Theorem:main_theorem},
as a function of the noise parameter $q$.
}
\label{Figure:L_n}
\end{figure}

We proceed to our main result on the entanglement-assisted covert capacity $C_{\text{cov-EA}}$ of the  depolarizing channel.
From this point on, we focus on Scenario 3, where Willie receives the first qubit of the environment (see Section~\ref{Subsection:channel_model}).
\begin{theorem}
    Consider a qubit depolarizing  channel $\mathcal{N}_{A \to BW}$
    as specified in Section~\ref{Subsection:channel_model} above, where $W=E_1$. The entanglement-assisted covert capacity  is bounded %
    as
    \begin{align}
        C_{\text{cov-EA}}(\mathcal{N}) \geq \frac{4\sqrt{2}}{3} \frac{(1-q)^2 }{(2-q)\sqrt{\eta(\omega_1||\omega_0)}}
    \end{align}
    where $\omega_0 \equiv \mathcal{N}_{A\to W}(\ketbra{0})$ and $\omega_1 \equiv \mathcal{N}_{A\to W}(\ketbra{1})$.
    \label{Theorem:main_theorem}
\end{theorem}
Note that $\eta(\omega_1||\omega_0)$ is defined in \eqref{eq:eta}. 
Our lower bound is depicted in Figure \ref{Figure:L_n}.
 As can be seen in the figure, our lower bound has the expected behavior for the covert capacity in the boundary points (see Remark \ref{Remark:bounds_of_q}). 
 For $q=0$, we have $C_{\text{cov-EA}}(\mathcal{N})=+\infty$ in the $\sqrt{n}\log{n}$ scale, because the warden only receives noise and Alice can transmit a linear number of information bits (effectively, there is no warden). Whereas, for $q=1$, the covert and non-covert capacities are zero. %

Following the  definitions in Section~\ref{Subsection:Rate}, a bound of the form $C_{\text{cov-EA}}\geq L_0$ implies that it is possible to  transmit
$ L_0 \sqrt{\delta n}\log{n}$ information bits reliably and covertly (see Definitions \ref{Definition:Achievable_Rate} and \ref{Definition:Capacity}). %
Recall that
without entanglement assistance, covert communication requirements limit the message to 
$O(\sqrt{n})$ information bits \cite{Bash-Quantum, Bash-Quantum-revision}.
Thereby, we have established that entanglement assistance increases the message scale in covert communication, from $O(\sqrt{n})$ to $O(\sqrt{n}\log{n})$ information bits.
A similar result has been shown for continuous-variable bosonic channels by
Gagatsos et al. 
\cite{Bash-Bosonic}.
To the best of our knowledge,
our result in Theorem~\ref{Theorem:main_theorem}, on the depolarizing channel, is
  the first demonstration of such a property for a finite-dimensional channel.

\begin{remark}
%
 In some communication settings, the coding scale is larger for continuous-variable channels. For example, in deterministic identification, the code size is super-exponential and scales as $2^{n \log{n} R}$ for Gaussian channels \cite{deterministic_identification_Gaussian_Salariseddigh} and Poisson channels \cite{salariseddigh2022deterministic_poisson}. On the other hand, deterministic identification is limited to an exponential scale for finite-dimensional channels \cite{Ahlswede1989IdentificationVC}. 
Nevertheless, we show here that in covert communication over a qubit depolarizing channel,  entanglement assistance can increase the number of information bits from $O(\sqrt{n} ) $ to $O(\sqrt{n} \log{n}) $, as in the bosonic case. In other words, the $\log{n} $ performance boost is not reserved to continuous variable systems.
\end{remark}

\section{Proof of Theorem \ref{Theorem:main_theorem}}
\subsection{Proof Idea}
Consider Scenario 3 presented in Section~\ref{Subsection:channel_model}. 
First, we identify an entangled state that meets the above condition for covertness.
As opposed to previous works \cite{Bloch,Bash-Quantum,wang2016fundamental}, 
we do \emph{not} encode a random bit sequence with $\sim\nicefrac{1}{\sqrt{n}}$ frequency (or probability) of $1$'s. Instead, we encode ``weakly" entangled states as in (\ref{eq:Superposition}),
such that the squared amplitude of this  quantum superposition of states describing innocent and non-innocent symbols is $\alpha=O\left(\nicefrac{1}{\sqrt{n}}\right)$.
 In order to guarantee covertness, the probability amplitude must be such that the state of the transmission is very close to that of a sequence of innocent states $\ket{0}^{\otimes n}$. 
Furthermore, we adapt the approach in \cite{Bash-Bosonic} %
to
analyze the order of the number of covert information bits.

\subsection{Position-Based Coding}
 
The lemma below provides an achievability result for the transmission over a memoryless quantum channel, regardless of covertness. 
%
%
%
For every $\rho,\sigma\in\mathscr{S}(\mathcal{H})$, define the second and fourth moments of the quantum relative entropy,
\begin{align}
     V(\rho||\sigma) &= \trace[\rho|(\log(\rho)-\log(\sigma ) - D(\rho||\sigma)|^2 ]\,,
     \label{Equation:Second_moment}
     \\
     Q(\rho||\sigma) &= \trace[\rho|(\log(\rho)-\log(\sigma ) - D(\rho||\sigma)|^4 ]\,,
     \label{Equation:Fourth_moment}
\end{align}
respectively.
\begin{lemma}[Position-based coding  {\cite[Lemma~1]{Bash-Bosonic} \cite{anshu2017one,khabbazi2019union,wilde2017position}}]
    Consider a memoryless quantum channel $\mathcal{N}_{A \to B}$.
    For every pure entangled state $\ket{\psi_{A_1 A}}\in \mathcal{H}_{A_1}\otimes\mathcal{H}_{A}$,
    arbitrarily small $\varepsilon>0$, and sufficiently large $n$,
    there exists a coding scheme that employs  pre-shared entanglement resources %
    to transmit $\log(M)$ bits over $n$ uses of $\mathcal{N}_{A \to B}$ with %
    decoding error probability $\varepsilon$ %
    such that:
    \begin{align}
     \log(M) &\geq nD(\psi_{A_1 B}||\psi_{A_1}\otimes \psi_{B})
    \nonumber\\& 
    + \sqrt{nV(\psi_{A_1 B}||\psi_{A_1}\otimes \psi_{B})} \Phi^{-1}(\varepsilon)-C_n
    \label{Equation:num_of_data_bits1}
    \intertext{with}
    \psi_{A_1 B}&=(\mathrm{id}_{A_1}\otimes \mathcal{N}_{A\to B})(\psi_{A_1 A})
    \label{Equation:psi_A1_B}
    \intertext{and}%
    C_n &= \frac {\beta_{\text{B-E}}}{\sqrt{2\pi}} \frac{[Q(\psi_{A_1 B}||\psi_{A_1}\otimes \psi_{B})]^\frac{3}{4}}{V(\psi_{A_1 B}||\psi_{A_1}\otimes \psi_{B})} 
    \nonumber\\&  
    + \frac{V(\psi_{A_1 B}||\psi_{A_1}\otimes \psi_{B})}{\sqrt{2\pi}} + \log(4\varepsilon n)
    \label{Equation:C_n}
    \end{align}
    where    $D(\cdot||\cdot)$ is the quantum relative entropy, $V(\cdot||\cdot)$, $Q(\cdot||\cdot)$ are the second and fourth moments in (\ref{Equation:Second_moment})-(\ref{Equation:Fourth_moment}), $\beta_{\text{B-E}}$ is the \emph{Berry-Esseen} constant satisfying $0.40973 \leq \beta_{\text{B-E}} \leq 0.4784$,
    and 
    \begin{align}
   \Phi^{-1}(\varepsilon) = \sup \{ \varepsilon \in [0,1] |  \Phi(\varepsilon) \leq \varepsilon \}\,, 
   \\
   \intertext{where}
   \Phi(\varepsilon) = \frac{1}{\sqrt{2\pi}}\int_{-\infty}^{\varepsilon}e^{\frac{x^2}{2}}dx\,.
\end{align}
    \label{Lemma:num_of_data_bits}
\end{lemma}
The derivation of Lemma~\ref{Lemma:num_of_data_bits} builds upon a \emph{position-based} coding scheme, where 
each message is associated with $n$ entangled pairs and Bob uses sequential decoding on the output and the entanglement resources for each message consecutively \cite{wilde2017position,Bash-Bosonic} (see proof of Lemma 1 in \cite{Bash-Bosonic}).

\subsection{Analysis}
\label{Subsection:Analysis}
In this section, we give the proof for Theorem \ref{Theorem:main_theorem}.
We present the main stages of the proof, while the technical details are deferred to the appendix.
We begin with the following lemma.
\begin{lemma}
\label{Lemma:Main_Technical}
Let $\gamma_n=n^{\nu-\frac{1}{6}}$, where $0<\nu<\frac{1}{6}$ is arbitrary and does not depend on $n$. 
    Then, there exists an entanglement-assisted covert coding scheme for qubit depolarizing channel with blocklength $n$, size $M$, and average error probability $\varepsilon$ that satisfies
    \begin{align}
        \log(M) \geq 2\left (\frac{2}{3} - \nu \right )\frac{(1-q)^2}{2-q} \gamma_n \sqrt{n} \log{n} + O(\sqrt{n}\gamma_n) \,.
        \label{Equation:lemma_for_main_theorem}
    \end{align}
\end{lemma}
\begin{proof} 
To prove the lemma, we need to show that,
for arbitrarily small $\varepsilon,\delta>0$ and  large $n$, there exists an $(M,n,\varepsilon,\delta)$  code for  the depolarizing channel with entanglement assistance, 
with a code size $M$ as in (\ref{Equation:lemma_for_main_theorem}).
To this end, we apply Lemma~\ref{Lemma:num_of_data_bits} with %
$\ket{\psi_{A_1 A}}$ as in (\ref{eq:Superposition}), with a parameter $\alpha=\alpha_n$ as in (\ref{Equation:alpha_definition}). 
Note that setting $\gamma_n=n^{\nu-\frac{1}{6}}$ as in the lemma statement yields
 \begin{align}
 \alpha_n=\frac{\gamma_n}{\sqrt{n}}=n^{\nu-\frac{2}{3}}.
\label{Equation:alpha_n_nu6}
 \end{align}
Intuitively,
as the value of $\alpha_n$ is  small,  the input state that Alice sends through the channel is close to the innocent state, i.e.,  $\psi_A \approx \ketbra{0}$.
Given the joint state $\psi_{A_1 A} \equiv \ketbra{\psi_{A_1 A}} $, the channel input $A$  is in the reduced state $\psi_A\equiv \Tr_{A_1}\left[\ketbra{\psi_{A_1 A}}\right]=\varphi_{\alpha_n}$, with $\varphi_{\alpha_n}$ as in (\ref{Equation:Bash_Quantum_State}).
That is, the reduced input state fits the achievability proof for the covert capacity without entanglement assistance in  \cite{Bash-Quantum, Bash-Quantum-revision}. %
Based on the analysis therein, this input state meets the covertness requirement.  %
As the covertness requirement does not involve the entanglement resources, it follows that covertness holds here as well, i.e., $D(\overline{\rho}_{W^n}||\omega_0^{\otimes n})$ tends to zero
as $n\to \infty$.

Having established both reliability and covertness, it remains to estimate the code size. To this end, consider
the joint state $\psi_{A_1 B}$ of the output system $B$ and the reference system $A_1$, as in (\ref{Equation:psi_A1_B}).
In order to
estimate each term on the right-hand side of (\ref{Equation:num_of_data_bits1}),
we first  derive expressions for 
the operator logarithms, $\log(\psi_{A_1 B})$  and $\log(\psi_{A_1}\otimes \psi_{B})$, and then 
we approximate the relative entropy
$ D(\psi_{A_1 B} || \psi_{A_1}\otimes \psi_{B})$,
and its second and fourth moments
$ V(\psi_{A_1 B} || \psi_{A_1}\otimes \psi_{B})$
and
$ Q(\psi_{A_1 B} || \psi_{A_1}\otimes \psi_{B})$.

The full technical details are given in the appendices. In Appendix \ref{appendix:log_calc},
we analyze the spectral decompositions, and then use the Taylor expansions   near $\alpha = 0$. Throughout the derivation, we maintain the exact value of the dominant terms and reduce the approximation error to its order class, following the asymptotic notation  in Section \ref{sec:notation}.  
In Appendix \ref{appendix:moments_calc},
we estimate the quantum relative entropy  and its  moments, and show that
\begin{align}
    D(\psi_{A_1 B} || \psi_{A_1}\otimes \psi_{B}) &= -2\frac{(1-q)^2}{2-q}\alpha_n\log(\alpha_n) + O(\alpha_n) \,,
    \nonumber\\
    V(\psi_{A_1 B} || \psi_{A_1}\otimes \psi_{B}) &= O\big(\alpha_n\log^2(\alpha_n)\big) \,,
    \nonumber\\
    Q(\psi_{A_1 B} || \psi_{A_1}\otimes \psi_{B}) &= O\big(\alpha_n\log^2(\alpha_n)\big) \,,
\end{align}
for $\alpha=\alpha_n$ as chosen above (see  (\ref{Equation:alpha_n_nu6})).

The proof is concluded
by placing the approximations above into (\ref{Equation:num_of_data_bits1}), as detailed in Appendix \ref{appendix:Rate_calc}. 
\end{proof}

We are now ready for the proof  of Theorem \ref{Theorem:main_theorem}.
\begin{proof}[Proof of Theorem \ref{Theorem:main_theorem}]

First, we observe that in this scenario, $\mathrm{supp}(\omega_1)\subseteq \mathrm{supp}(\omega_0)$ and, in addition, $\omega_0 \neq \omega_1$ (see derivation in Appendix \ref{Subsection:willie_recives_first_qubit}), therefore,
covert communication is possible and not trivial. Then, even if  Willie's output state is $\omega_1$, there is still ambiguity whether the input is innocent or not.

By Lemma~\ref{Lemma:Main_Technical}, we have established achievability for the following covert rate:
\begin{align}
    L_n &=   \frac{2 \left (\frac{2}{3} - \nu \right )\frac{(1-q)^2}{2-q} \gamma_n + O\left(\frac{\gamma_n}{\log{n}}\right)}{\sqrt{D(\overline{\rho}_{W^n}||\omega_0^{\otimes n}})}  \,.
    \label{Equation:covert_rate_theorem}
\end{align}
We have seen that 
covertness holds as 
the reduced input state 
is the same as the average input in previous code constructions 
\cite{Bash-Quantum, Bash-Quantum-revision}.
Furthermore, the following property extends as well:  
there exists $\zeta_n>0$ such that,
\begin{align}
    |D(\overline{\rho}_{W^n}||\omega_0^{\otimes n})-nD(\omega_{\alpha_n}||\omega_{0})| \leq e^{-\zeta_n\gamma_n^\frac{3}{2}n^\frac{1}{4}}\,,
    \label{Equation:Divergence_cqChannel}
\end{align}
where
$\overline{\rho}_{W^n}$ is the actual state of Willie's system as defined in
(\ref{Equation:Willie_Average_State}), the state
$\omega_0=\mathcal{N}_{A\to W}(\ketbra{0})$ is the Willie's output corresponding to the innocent input, and
$\omega_{\alpha_n} \equiv \mathcal{N}_{A \to W}(\varphi_{\alpha_n})$, with $\varphi_{\alpha_n}$ as in (\ref{Equation:Bash_Quantum_State})
and $\zeta_n\in \omega\left(\frac{1}{\log^2 n} \right)$
(see achievability proof in \cite{Bash-Quantum}, \cite[Theorem~1]{Bash-Quantum-revision}).
This holds since the derivation depends on the reduced input state alone, as
Willie does not have access to the entanglement resource. 

Based on a result that was recently developed for covert sensing using entangled states \cite[Lemma~5]{tahmasbi2021covert}, 
\begin{align} 
D(\omega_{\alpha_n}||\omega_{0}) = \frac{\alpha_n^2}{2}\eta(\omega_1||\omega_0) + O(\alpha_n^3)
 \label{Equation:Divergence_Sensing}
\end{align}
 for sufficiently small $\alpha_n$.
Thus, by (\ref{Equation:Divergence_cqChannel}) and (\ref{Equation:Divergence_Sensing}),
\begin{align}
    D(\overline{\rho}_{W^n}||\omega_0^{\otimes n}) \leq \frac{\gamma_n^2}{2}\eta(\omega_1||\omega_0)  + e^{-\zeta\gamma_n^\frac{3}{2}n^\frac{1}{4}} 
    +O\left(\frac{\gamma_n^3}{\sqrt{n}} \right)\,.
\end{align}

By applying this bound to the denominator in  (\ref{Equation:covert_rate_theorem}), we have:
\begin{align}
    L_n \geq   \frac{2 \left (\frac{2}{3} - \nu \right )\frac{(1-q)^2}{2-q} \gamma_n + O\left(\frac{\gamma_n}{\log{n}}\right)}{\sqrt{\frac{\gamma_n^2}{2}\eta(\omega_1||\omega_0)  + e^{-\zeta\gamma_n^\frac{3}{2}n^\frac{1}{4}} 
    +O\left(\frac{\gamma_n^3}{\sqrt{n}} \right)}}
    \,.
\end{align}
Hence, in
 the limit of $n \to \infty$, we achieve 
 \begin{align}
    L
    \geq   
    \frac{2 \left (\frac{2}{3} - \nu \right )\frac{(1-q)^2}{2-q} }{\sqrt{\frac{1}{2}\eta(\omega_1||\omega_0)}} 
 \end{align}
for arbitrarily small $\nu>0$, which completes the proof.
\end{proof}

\section{Energy Constraint Interpretation}
\label{Section:Energy Constraint Interpretation}
We provide an interpretation of the logarithmic advantage. 
 %
In the bosonic case, %
the ratio between the entanglement-assisted capacity and the unassisted capacity 
follows a logarithmic trend of 
$\log(1/E)$, where $E$ is the limit on the transmission mean photon number \cite{holevo2003entanglement,9173940}.
Yet, to ensure covertness, the mean photon number must be restricted to %
$E_n=O(\frac{1}{\sqrt{n}})$. 
Consequently, an $%
O(\log{n})$ factor arises \cite{Shi_2020}.
Based on our derivation,  a similar phenomenon is observed for the qubit depolarizing channel.

Indeed, consider communication over a finite-dimensional channel under an energy constraint, $E$, without the covertness constraint \cite[Sec. 2]{holevo2003entanglement}. 
Then, the capacities with and without entanglement assistance, are given by \cite{holevo2003entanglement}
\begin{align}
C_0(\mathcal{N},E)&=\max_{\{ p_x(x), \phi_A^{(x)} \}: \trace(\mathsf{F}\rho_A)\leq E } I(X;B)_\rho
\label{Eqution:Unassisted_Capacity}
\\
C_{\text{EA}}(\mathcal{N},E)&=\max_{ \psi_{A_1 A} : \trace(\mathsf{F}\psi_A)\leq E } I(A_1;B)_\omega
\label{Equation:EA_Capacity}
\end{align}
with the observable (Hamiltonian) 
$\mathsf{F}=\ketbra{1}$, where
\begin{align}
\rho_{XA}&=\sum_{x\in\mathcal{X}} p_X(x) \ketbra{x}\otimes \phi_A^{(x)}
 \,,
\\
\rho_{XB}&=(\mathrm{id}_X\otimes \mathcal{N}_{A\to B})(\rho_{XA}) \,,
\intertext{and}
\omega_{A_1 B}&=(\mathrm{id}_X\otimes \mathcal{N}_{A\to B})(\psi_{A_1 A})
\end{align}
The maximization in 
\eqref{Eqution:Unassisted_Capacity} is over all the input ensembles $\{ p_x(x), \phi_A^{(x)}\}$ such that the reduced average state $\rho_{A}\equiv \trace_X(\rho_{XA})$ satisfies the energy constraint
$\trace(\mathsf{F}\rho_A)\leq E$.
Similarly, the maximization in \eqref{Equation:EA_Capacity} is over all the entangled input states $\ket{\psi_{A_1 A}}$ with a reduced state $\psi_A$ such that 
$\trace(\mathsf{F} \psi_A)\leq E$.


Now, consider the qubit depolarizing channel with an energy constraint $E$, where $0<E\leq \frac{1}{2}$.
Without assistance, the ensemble that achieves the maximum is 
 $\left\{ \left( 1-E,E \right), \ket{0},\ket{1}  \right\}$.
The capacity without entanglement assistance is thus given by 
\begin{align}
C_0(\mathcal{N},E)&=h_2\left(E*\frac{q}{2}\right)-h_2\left(\frac{q}{2}\right) \,,
\label{Equation:Energy_Capacity_0}
\end{align}  
%
where `$*$' denotes the binary convolution operation: $\alpha*\beta=(1-
\alpha)\beta+\alpha(1-\beta)$.

As for the entanglement-assisted capacity, the maximum is attained for
 \begin{align}
\ket{\Psi_{A_1 A}} &\equiv \sqrt{1-E}\ket{00} + \sqrt{E}\ket{11} \,.
\end{align}
 %
Therefore,
\begin{align}
 C_\text{EA}(\mathcal{N},E)&=
 h_2(E)+h_2\left(E*\frac{q}{2}\right)-H(\psi_{A_1 B}) 
 \,.
\end{align} 
where
\begin{align}
\psi_{A_1 B}&= (\mathrm{id}\otimes\mathcal{N}_{A\to B})(\ketbra{\Psi_{A_1 A}}) \,.
\end{align}
For completeness, we prove the capacity characterizations above in Appendix~\ref{Appendix:Energy_Constraint}.

Now, based on the derivations in Subsection~\ref{Subsection:Analysis}, 
for $E \to 0$,   
 we have
\begin{align}
    \frac{ C_{\text{EA}}(\mathcal{N},E)}{C_{0}(\mathcal{N},E)} &= \frac{h_2(E)+h_2\left(E*\frac{q}{2}\right)-H(\psi_{A_1 B})}{h_2\left(E*\frac{q}{2}\right)-h_2\left(\frac{q}{2}\right)} 
    \nonumber\\
    &\sim\frac{-E\log(E)}{E} 
    \nonumber\\
    &= -\log(E) \,,
\end{align}
by taking $\alpha=E$.
To satisfy  the covert constraint, we effectively 
impose an energy constraint $E_n \sim \frac{1}{\sqrt{n}}$, which results in the following ratio between the entanglement-assisted and unassisted covert capacities, 
\begin{align}
    \frac{ C_{\text{EA-cov}}(\mathcal{N})}{C_{0\text{-cov}}(\mathcal{N})} \sim \log{n} \,.
\end{align}

\section{Summary and Discussion}
We have studied covert communication through the qubit depolarizing channel, where Alice and Bob share entanglement resources and wish to communicate, while an adversarial warden, Willie, tries to detect their communication. We addressed 
three scenarios. In the first scenario, Willie can determine with certainty whether  Alice has transmitted a non-innocent state, making covert communication impossible. In the second, Willie cannot distinguish between the $\ket{0}$ and $\ket{1}$ inputs, making covert communication effortless.
The outcomes of our study mainly pertain to the third scenario, wherein covert communication is both feasible and non-trivial.
Our results show that it is possible to transmit $O(\sqrt{n}\log{n})$  bits reliably and covertly. This result surpasses the maximum scaling of $O(\sqrt{n})$ reliable and covert bits in both the classical and  quantum cases without entanglement assistance. %

The square root law for the unassisted cases (both classical and quantum)  was derived for the non-trivial scenario,  in which Bob cannot determine with certainty if Alice sends a non-innocent symbol. However, if Bob has this capability, i.e., $\mathrm{supp}(\mathcal{N}_{A \to B} (\ketbra{1}))  \nsubseteq \mathrm{supp}(\mathcal{N}_{A \to B} (\ketbra{0}))$, then the scaling law becomes $O(\sqrt{n}\log{n})$, even for a classical channel \cite{Bash-Quantum, Bash-Quantum-revision,Bloch}. 
Therefore, it appears that entanglement assistance has a similar effect as granting Bob the capability of identifying a non-innocent transmission with certainty.
We also discussed the energy constraint interpretation in Section~\ref{Section:Energy Constraint Interpretation}, where we have seen that the entanglement-assisted and unassisted capacities under an energy constraint scale as 
$C_{\text{EA}}(\mathcal{N},E)\sim -E\log(E)$ and 
$C_{0}(\mathcal{N},E)\sim E$, respectively, without covertness.
Hence, the ratio between those capacities follows $\log(1/E)$.
The covertness constraint effectively imposes an energy constraint of $E_n\sim \frac{1}{\sqrt{n}}$.
Hence, the ratio between the covert entanglement-assisted and unassisted capacity scales as $\log{n}$.
While the energy constraint interpretation provides another view on this behavior, a full understanding of the effect of entanglement resources on the performance remains elusive. 

A promising future research direction is to consider a more general model, where the covert communication channel is formed by a concatenation of the depolarizing channel $\mathcal{V}_{A\to BE}$ with a general channel $\mathcal{P}_{E\to W}$ to Willie, namely,
$\mathcal{N}_{A\to BW}=(\mathrm{id}_B\otimes \mathcal{P}_{E\to W})\circ \mathcal{V}_{A\to BE}$. Those scenarios are out of scope for the current paper, but it would be interesting to consider in future work.
The amount of entanglement utilized also requires further study. 
Recently,  Wang et al. \cite{bloch_resource_efficient} improved the previous result by Gagatsos et al. \cite{Bash-Bosonic} and showed achievability using $\sim\sqrt{n}$ two-mode squeezed vacuum states, i.e., entanglement of dimension $\sim 2^{\sqrt{n}}$, which is negligible when compared to the code size.
Here, relying on position-based coding, we use 
$n$ qubit pairs per message. It would be worthwhile to explore methods to reduce the entanglement dimension within the finite dimensional setting as well.

Our results can be viewed as a step forward towards understanding covert communication via general quantum channels in the presence of pre-shared entanglement resources. Following the past literature, the preliminary results on entanglement-assisted communication via the depolarizing and erasure channels \cite{bennett1999entanglement} have led to a complete characterization for a general quantum channel \cite{EA_capacity_bennet}. We can only hope to see the same progress in the study of covert communication. 
The quantum erasure channel is another fundamental model in quantum information theory
\cite{bennett1997capacities}, 
 where for an input state $\rho$, Bob receives the original state with probability $1-q$, or an erasure state $\ketbra{e}$, which is orthogonal to the qubit space, with probability $q$. For this channel, Bob can determine that Alice sent $\ket{1}$ with certainty, as $\mathrm{supp}(\mathcal{N}_{A \to B} (\ketbra{1}))  \nsubseteq \mathrm{supp}(\mathcal{N}_{A \to B} (\ketbra{0}))$.
 Thereby,  the scaling law  becomes $O(\sqrt{n}\log{n})$ information bits, even without entanglement resources.
At this point, it remains unclear whether 
this 
scaling can be achieved with entanglement assistance for every quantum  channel that satisfies   $\mathrm{supp}(\mathcal{N}_{A \to W} (\ketbra{1}))  \subseteq \mathrm{supp}(\mathcal{N}_{A \to W} (\ketbra{0}))$.

\section*{Acknowledgment}
%
The authors would like to thank Johannes Rosenberger (Technical University of Munich) for useful discussions. Furthermore, the authors thank the anonymous reviewer for pointing out the interpretation and relation to the asymptotic behavior of the ratio between the entanglement-assisted and unassisted capacities.

\section*{Appendix Organization}
The appendices are organized as follows.
In Appendix \ref{appendix:willies_channel} we provide the technical analysis of the channel from Alice to Willie.
Appendix \ref{appendix:log_calc} presents mathematical tools and derivations for decomposing the operators %
$\psi_{A_1 B}$ and $\psi_{A_1} \otimes \psi_B$, and their logarithms.
In Appendix \ref{appendix:moments_calc}, we provide the detailed  approximation of  $D(\psi_{A_1 B}||\psi_{A_1}\otimes \psi_{B})$ and its moments. Appendix \ref{appendix:Rate_calc} presents the approximation of the code size.

 \begin{appendices}

 \section{Willie's Channels}
\label{appendix:willies_channel}
\subsection
[Willie Receives E1 E2] 
{Willie Receives $(E_1,E_2)$} 
\label{Subsection:willie_recives_all_E}
For the given scenario where Willie receives the entire environment, it is possible to demonstrate that,
\begin{align}
   \omega_0 =
    \begin{pmatrix}
        1-\frac{3q}{4} & 0 & 0 & \sqrt{\frac{q}{4}(1-\frac{3q}{4})}\\
        0 & \frac{q}{4} & -i\frac{q}{4}& 0\\
        0 & i\frac{q}{4} & \frac{q}{4} & 0\\
        \sqrt{\frac{q}{4}(1-\frac{3q}{4})} & 0 & 0 & \frac{q}{4}
    \end{pmatrix} \,,
\end{align}
and
\begin{align}
   \omega_1 =
    \begin{pmatrix}
        1-\frac{3q}{4} & 0 & 0 & -\sqrt{\frac{q}{4}(1-\frac{3q}{4})}\\
        0 & \frac{q}{4} & i\frac{q}{4}& 0\\
        0 & -i\frac{q}{4} & \frac{q}{4} & 0\\
        -\sqrt{\frac{q}{4}(1-\frac{3q}{4})} & 0 & 0 & \frac{q}{4}
    \end{pmatrix} \,.
\end{align}

The null spaces of $\omega_0$ and $\omega_1$ contain vectors,
\begin{align}
        \ket{e_0} \equiv  \begin{pmatrix}
        0\\
        i\\
        1\\
        0
    \end{pmatrix} \,,
\end{align}
and
\begin{align}
        \ket{e_1} \equiv  \begin{pmatrix}
        0\\
        -i\\
        1\\
        0
    \end{pmatrix} \,,
\end{align}
respectively.
Since $\braket{e_0}{e_1} = 0$, it follows that $\mathrm{supp}(\omega_1)\not\subseteq \mathrm{supp}(\omega_0)$.

\subsection
[Willie Receives E2] 
 {Willie Receives $E_2$} 
\label{Subsection:willie_recives_second_qubit}
Suppose Alice transmits the general state $\rho = (1-a)\ketbra{0} + a \ketbra{1} + b \ketbra{0}{1} + b^* \ketbra{1}{0}$. Then, Willie receives the state,
\begin{align}
     \mathcal{N}_{A\to W}(\rho) &=  \left(1-\frac{q}{2}\right)\ketbra{0} + \frac{q}{2} \ketbra{1}
    \nonumber\\& 
    \phantom{=}+ 2\Re\{b\} \bigg ( \left(\sqrt{\left(1-\frac{3q}{4}\right)\frac{q}{4}} + i\frac{q}{4}\right)\ketbra{0}{1} 
    \nonumber\\&
    \phantom{=}+ \left(\sqrt{\left(1-\frac{3q}{4}\right)\frac{q}{4}} - i\frac{q}{4}\right)\ketbra{1}{0} \bigg ) \,.
    \label{Eq:willie_state_E_2}
\end{align}
Substituting $\rho = \ketbra{0}$ and $\rho = \ketbra{1}$ into \eqref{Eq:willie_state_E_2}, respectively, yields:
\begin{align}
    \omega_0 &= \mathcal{N}_{A\to W}(\ketbra{0}) 
    \nonumber\\
    &= \left(1-\frac{q}{2}\right)\ketbra{0} + \frac{q}{2} \ketbra{1} \,,
\end{align}
and
\begin{align}
    \omega_1 &= \mathcal{N}_{A\to W}(\ketbra{1}) 
    \nonumber\\
    &= \left(1-\frac{q}{2}\right)\ketbra{0} + \frac{q}{2} \ketbra{1} \,.
\end{align}

\subsection
[Willie Receives E1] 
{Willie Receives $E_1$} 
\label{Subsection:willie_recives_first_qubit}
Suppose Alice transmits the general state $\rho = (1-a)\ketbra{0} + a \ketbra{1} + b \ketbra{0}{1} + b^* \ketbra{1}{0}$. Then, Willie receives the state,
\begin{align}
     \mathcal{N}_{A\to W} &=  \left(1-\frac{q}{2}\right)\ketbra{0} + \frac{q}{2} \ketbra{1}
    \nonumber\\& 
    \phantom{=}+ (1-2a) \bigg ( \left(\sqrt{\left(1-\frac{3q}{4}\right)\frac{q}{4}} - i\frac{q}{4}\right)\ketbra{0}{1} 
    \nonumber\\&
    \phantom{=}+ \left(\sqrt{\left(1-\frac{3q}{4}\right)\frac{q}{4}} + i\frac{q}{4}\right)\ketbra{1}{0} \bigg )
    \label{Eq:willie_state_E_1}
\end{align}
Substituting $\rho = \ketbra{0}$ and $\rho = \ketbra{1}$ into \eqref{Eq:willie_state_E_1}, respectively, yields:
\begin{align}
    \omega_0 &= \mathcal{N}_{A\to W}(\ketbra{0}) 
    \nonumber\\
    &= \left(1-\frac{q}{2}\right)\ketbra{0}
+
\frac{q}{2}\ketbra{1}
+
\left(\sqrt{1-\frac{3q}{4}}\sqrt{\frac{q}{4}} + i\frac{q}{4}\right) \ketbra{1}{0}
+
\left(\sqrt{1-\frac{3q}{4}}\sqrt{\frac{q}{4}} - i\frac{q}{4}\right) \ketbra{0}{1} \,,
\end{align}
and
\begin{align}
    \omega_1 &= \mathcal{N}_{A\to W}(\ketbra{1}) 
    \nonumber\\
    &= \left(1-\frac{q}{2}\right)\ketbra{0}
+
\frac{q}{2}\ketbra{1}
-
\left(\sqrt{1-\frac{3q}{4}}\sqrt{\frac{q}{4}} + i\frac{q}{4}\right) \ketbra{1}{0}
-
\left(\sqrt{1-\frac{3q}{4}}\sqrt{\frac{q}{4}} - i\frac{q}{4}\right) \ketbra{0}{1} \,.
\end{align}

The determinant of both $\omega_0$ and $\omega_1$ is,
\begin{align}
    |\omega_0| &= |\omega_1| = \frac{3q}{8}(2-q) \,.
\end{align}
Since the determinant is not equal to zero (for $0<q<1$), it follows that $\omega_0$ and $\omega_1$ span the entire qubit space, thus, in particular $\mathrm{supp}(\omega_1)\subseteq \mathrm{supp}(\omega_0)$.

 \section{Matrix Logarithms Estimation}
 \label{appendix:log_calc}

\subsection{Approximation Tools}
%
We  provide the  approximation tools that are used throughout the derivation, using the 
``big $\mathcal{O}$-notation" in 
Section~\ref{sec:notation}.
\begin{itemize}
\item Useful Taylor expansions (at $x=0$):
\begin{align}
    \sqrt{a+bx+cx^2} &= \sqrt{a} + \frac{b}{2\sqrt{a}}x + \mathcal{O}(x^2) \,,
    \label{Equation:sqrt_series}
    \\
    \log(a+bx+cx^2) &= \frac{\ln(a)}{\ln(2)}+ \frac{b}{a\ln(2)}x 
    \nonumber\\& \phantom{=}
    -\frac{b^2-2ac}{a^2\ln(4)}x^2 + \mathcal{O}(x^3) \,,
    \label{Equation:log2_series}
    \\
    \sqrt{x(1-x)} &= \sqrt{x} + \mathcal{O}(x^{\frac{3}{2}}) \,,
    \label{Equation:nondiagonal_series}
    \\
    \frac{x}{\sqrt{x(1-x)}} &= \sqrt{x} + \mathcal{O}(x^{\frac{3}{2}}) \,,
    \label{Equation:lambda4_factor_series}
    \\
    \frac{1}{\sqrt{x(1-x)}} &= \frac{1}{\sqrt{x}} + \mathcal{O}(\sqrt{x}) \,,
    \label{Equation:lambda1_factor_series}
    \\
    \frac{1}{\sqrt{cx+1}} &= 1+ \mathcal{O}(x) \,,
    \label{Equation:lambda4_normalization_series}
    \\
    \frac{1}{\sqrt{\frac{c}{x}+1}} &= \frac{1}{\sqrt{c}}\sqrt{x} + \mathcal{O}(x^\frac{3}{2}) \,.
    \label{Equation:lambda1_normalization_series}
\end{align}
\item The spectral decomposition of a Hermitian operator,
\begin{align}
    P &= a\ketbra{00} + b\ketbra{01} + c\ketbra{10} + d\ketbra{11} 
    \nonumber\\&
    \phantom{=}
    + s(\ketbra{00}{11} + \ketbra{11}{00}) \,
\end{align}
consists of the eigenvalues
\begin{align}
    &\lambda_{1} = \frac{1}{2}\left ( a+d + \sqrt{(a+d)^2-4(ad-s^2)}\right ),
    \nonumber\\
    &\lambda_{4} = \frac{1}{2}\left ( a+d - \sqrt{(a+d)^2-4(ad-s^2)}\right ),
    \nonumber\\
    &\lambda_2 = b,
    \nonumber\\
    &\lambda_3 = c \,.
     \label{Equation:eigenvalues_of_decomposition}
\end{align}
and the associated eigenvectors,
\begin{align}
    &\ket{\lambda_{1}} = C_{1}\left (\widetilde{\lambda}_{1}\ket{00}+\ket{11} \right ),
    \nonumber\\
    &\ket{\lambda_{4}} = C_{4}\left (\widetilde{\lambda}_{4}\ket{00}+\ket{11} \right ),
    \nonumber\\
    &\ket{\lambda_{2}} = \ket{01},
    \nonumber\\
    &\ket{\lambda_3} = \ket{10} \,.
    \label{Equation:eigenvectors_of_decomposition}
\end{align}
where
\begin{align}
    \widetilde{\lambda}_{1} &\equiv -\frac{a-\lambda_{1}}{s} \,,
    &
    \widetilde{\lambda}_{4} \equiv -\frac{a-\lambda_{4}}{s} \,,
    \\
    C_{1} &\equiv \frac{1}{\sqrt{\widetilde{\lambda}_{1}^2 +1}} \,,&
     C_{4} \equiv \frac{1}{\sqrt{\widetilde{\lambda}_{4}^2 +1}} \,.
\end{align}
\end{itemize}

\subsection{Output density operators}
The joint state $\psi_{A_1 B}$ of the reference system and Bob's output  is obtained by applying the depolarizing channel: %
\begin{align}
&
    \psi_{A_1 B} = (\mathbb{1}_{A_1} \otimes \mathcal{N}_{A \to B})( \psi_{A_1 A})
    \nonumber\\
    &= \left(1-\frac{3}{4}q\right)\psi_{A_1 A}
    +\frac{q}{4} [(\mathbb{1}_{A_1} \otimes X)\psi_{A_1 A}(\mathbb{1}_{A_1} \otimes X) 
    \nonumber\\&  \phantom{=}
    + (\mathbb{1}_{A_1} \otimes Y)\psi_{A_1 A}(\mathbb{1}_{A_1} \otimes Y) 
    \nonumber\\
    &\phantom{=}+ (\mathbb{1}_{A_1} \otimes Z)\psi_{A_1 A}(\mathbb{1}_{A_1} \otimes Z) ] \,.
\end{align}
Algebraic manipulations yield
\begin{align}
    \psi_{A_1 B} 
    &= \left(1-\frac{q}{2}\right)(1-\alpha)\ketbra{00} + \left(1-\frac{q}{2}\right)\alpha\ketbra{11} 
    \nonumber\\& \phantom{=}
    +\frac{q}{2}(1-\alpha)\ketbra{01} + \frac{q}{2}\alpha\ketbra{10}
    \nonumber\\
    &\phantom{=}+ (1-q)\sqrt{\alpha}\sqrt{1-\alpha}(\ketbra{00}{11}+\ketbra{11}{00}) \,.
\end{align}
The reduced matrices $\psi_{A_1}$ and $\psi_{B}$ are, thus,
\begin{align}
    \psi_B %
    &= \left[\left(1-\frac{q}{2}\right)*\alpha\right] \ketbra{0}+
    \left[\frac{q}{2}*\alpha\right] \ketbra{1} \,,
    \label{Equation:reduced_matrix_of_out_depolarizing}
    \\
    \psi_{A_1} %
    &=  (1-\alpha)\ketbra{0}+\alpha\ketbra{1},
\end{align}
where $\alpha*\beta=(1-\alpha)\beta+\alpha(1-\beta)$.
%
Then, %
\begin{align}
&
    \psi_{A_1}\otimes \psi_{B} 
    \nonumber\\
    &= (1-\alpha)\left[\left(1-\frac{q}{2}\right)*\alpha\right]\ketbra{00} 
    \nonumber\\&\phantom{=}
    + (1-\alpha)\left[\frac{q}{2}*\alpha\right]\ketbra{01} 
    \nonumber\\&\phantom{=}
    + \alpha\left[\left(1-\frac{q}{2}\right)*\alpha\right]\ketbra{10}
    + \alpha\left[\frac{q}{2}*\alpha\right]\ketbra{11} \,.
\end{align}

The logarithm of $\psi_{A_1}\otimes \psi_{B}$ can be computed directly as it is diagonal in the standard basis. This is not the case for $\psi_{A_1 B}$.
Using (\ref{Equation:eigenvalues_of_decomposition}),  the spectral decomposition %
consists of the following eigenvalues:
\begin{align}
     \lambda_{1} &= \frac{1}{2} \left( 1-\frac{q}{2}
     %
     %
     +
     \sqrt{\left[1-\frac{q}{2}\right]^2-4q\left[1-\frac{3q}{4}\right]\alpha+4q\left[1-\frac{3q}{4}\right]\alpha^2} \right),
     \nonumber\\
     \lambda_{4} &= \frac{1}{2} \left( 1-\frac{q}{2}
     %
     -
     \sqrt{\left[1-\frac{q}{2}\right]^2-4q\left[1-\frac{3q}{4}\right]\alpha+4q\left[1-\frac{3q}{4}\right]\alpha^2} \right),
     \nonumber\\
    \lambda_2 &= \frac{q}{2}(1-\alpha) \text{, } 
    \lambda_3 = \frac{q}{2}\alpha \,.
\end{align}
Using the Taylor approximation in (\ref{Equation:sqrt_series}), we approximate $\lambda_1$ and $\lambda_4$ by
\begin{align}
    \lambda_{1} = \frac{1}{2} \left ( 1-\frac{q}{2} + \left( 1-\frac{q}{2}\right) - \frac{4q\left(1-\frac{3}{4}q\right)}{2\left(1-\frac{q}{2}\right)} \alpha + \mathcal{O}(\alpha^2) \right ) \,,
    \nonumber\\
    \lambda_{4} = \frac{1}{2} \left ( 1-\frac{q}{2} - \left( 1-\frac{q}{2}\right) - \frac{4q\left(1-\frac{3}{4}q\right)}{2\left(1-\frac{q}{2}\right)} \alpha + \mathcal{O}(\alpha^2) \right ) \,.
\end{align}
That is,
\begin{align}
    \lambda_1 &= 1-\frac{q}{2} - \frac{q\left(1-\frac{3}{4}q\right)}{\left(1-\frac{q}{2}\right)} \alpha +\mathcal{O}(\alpha^2) \,,
    \\
    \lambda_4 &=  \frac{q\left(1-\frac{3}{4}q\right)}{\left(1-\frac{q}{2}\right)} \alpha +\mathcal{O}(\alpha^2)\,.
\end{align}

The eigenvectors of $\psi_{A_1B}$ are given in \eqref{Equation:eigenvectors_of_decomposition}, with $\widetilde{\lambda}_1$ and $\widetilde{\lambda}_4$
satisfying
\begin{align}
        \widetilde{\lambda}_1%
    &= \frac{q-2}{2(q-1)\sqrt{\alpha}} + \mathcal{O}(\sqrt{\alpha})\,,
    \\
      \widetilde{\lambda}_4%
    &= -\frac{2(q-1)}{(q-2)}\sqrt{\alpha} + \mathcal{O}(\sqrt{\alpha}^\frac{3}{2}) \,,
\end{align}
by \eqref{Equation:lambda1_factor_series}. and
\begin{align}
    C_1^2 &= 4\frac{(q-1)^2}{(q-2)^2}\alpha + \mathcal{O}(\alpha^2)\,,\;
    C_4^2 
    = 1+ \mathcal{O}(\alpha) \,,
\end{align}
by \eqref{Equation:lambda1_normalization_series}.


By applying (\ref{Equation:log2_series}), we approximate the logarithm of the eigenvalues as follows. For the joint state $\psi_{A_1 B}$,
\begin{align}
     \log (\lambda_1) &= \log\left(1-\frac{q}{2}\right) + \mathcal{O}(\alpha) \,,
    \\
    \log (\lambda_2) &= \log(\frac{q}{2}) + \mathcal{O}(\alpha)\,,
    \\
    \log (\lambda_3) &= \log(\frac{q}{2}) + \log(\alpha) 
    + \mathcal{O}(\alpha^2) \,,
    \\
    \log (\lambda_4) &= \log(C(q)) + \log(\alpha) 
        + \mathcal{O}(\alpha^2)\,.
\end{align}
As for the product state $\psi_{A_1}\otimes \psi_B$, we have
\begin{align}
    \log((1-\alpha)\left[\left(1-\frac{q}{2}\right)*\alpha\right]) &= \log\left(1-\frac{q}{2}\right) + \mathcal{O}(\alpha)\,,
    \\
    \log((1-\alpha)\left[\left(\frac{q}{2}\right)*\alpha\right]) &= \log(\frac{q}{2}) + \mathcal{O}(\alpha)\,,
    \\
    \log((\alpha)\left[\left(1-\frac{q}{2}\right)*\alpha\right]) &= \log\left(1-\frac{q}{2}\right) 
    + \log(\alpha)
    + \mathcal{O}(\alpha)\,,
    \\
    \log((\alpha)\left[\left(\frac{q}{2}\right)*\alpha\right]) &= \log(\frac{q}{2}) + \log(\alpha). 
\end{align}
Hence, the operator-logarithm for $\psi_{A_1 B}$ satisfies
\begin{align}
     &\log(\psi_{A_1 B})
     \nonumber\\
     &= \log (\lambda_1) \ketbra{\lambda_1} + \log (\lambda_2) \ketbra{\lambda_2} 
     \nonumber\\& \phantom{=} 
     + \log (\lambda_3) \ketbra{\lambda_3} +  \log (\lambda_34) \ketbra{\lambda_4} 
    \nonumber\\
    &= \left[\left(1-\frac{q}{2}\right) + \frac{4(q-1)^2}{(q-2)^2}\alpha\log(\alpha) 
    + \mathcal{O}(\sqrt{\alpha}) \right]\ketbra{00}
    \nonumber\\
    &\phantom{=}+ \left[ \log(\frac{q}{2}) + \mathcal{O}(\alpha)\right]\ketbra{01}
    \nonumber\\
    &\phantom{=}+ \left[\log(\frac{q}{2}) + \log(\alpha) + \mathcal{O}(\alpha^2) \right] \ketbra{10} 
    \nonumber\\
    &\phantom{=}+ \left [ \log(C(q)) + \log(\alpha) + \mathcal{O}(\alpha\log(\alpha)) \right ]
    \ketbra{11}
    \nonumber\\
    &\phantom{=}+ \left [ -\frac{2(q-1)}{(q-2)}\sqrt{\alpha}\log(\alpha) + \mathcal{O}(\sqrt{\alpha}) \right ] 
    \ketbra{00}{11}
    \nonumber\\
    &\phantom{=}+ \left [ -\frac{2(q-1)}{(q-2)}\sqrt{\alpha}\log(\alpha) + \mathcal{O}(\sqrt{\alpha}) \right ] 
    \ketbra{11}{00}\,,
    \label{Equation:psi_AB}
\end{align}
and for  $\psi_{A_1}\otimes \psi_{B}$,
\begin{align}
&    \log(\psi_{A_1}\otimes \psi_{B}) 
\nonumber\\
&= \left[\log\left(1-\frac{q}{2}\right) + \mathcal{O}(\alpha)\right]\ketbra{00} 
    \nonumber\\
    &\phantom{=}
    + \left[\log(\frac{q}{2}) + \mathcal{O}(\alpha)\right]\ketbra{01}
    \nonumber\\
    &\phantom{=}+ \left[\log\left(1-\frac{q}{2}\right) + \log(\alpha) + \mathcal{O}(\alpha)\right]\ketbra{10} 
    \nonumber\\
    &\phantom{=}+ \left[\log(\frac{q}{2}) + \log(\alpha) + \mathcal{O}(\alpha)\right]\ketbra{11} \,.
    \label{Equation:psi_AB_product}
\end{align}

 \section{Relative Entropy and Moments}
 In this section, we develop the approximations for the relative entropy $D(\psi_{A_1B}||\psi_{A_1}\otimes \psi_{B})$, and its second and fourth moments,  $V(\psi_{A_1B}||\psi_{A_1}\otimes \psi_{B})$ and  $Q(\psi_{A_1B}||\psi_{A_1}\otimes \psi_{B})$.
 \label{appendix:moments_calc}
 \subsection{Relative Entropy}
 Consider the relative entropy, $D(\psi_{A_1B}||\psi_{A_1}\otimes \psi_{B})$.
By subtracting  (\ref{Equation:psi_AB_product}) from (\ref{Equation:psi_AB}), 
\begin{align}
    &%
    \log(\psi_{A_1 B}) - \log(\psi_{A_1}\otimes \psi_{B}) 
    \nonumber\\
   &=
   \left [
   \frac{4(q-1)^2}{(q-2)^2}\alpha\log(\alpha) + \mathcal{O}(\sqrt{\alpha})
    \right ]\ketbra{00} 
   +  \left [
   \mathcal{O}(\alpha )
   \right ]\ketbra{01} 
   \nonumber\\
   & \phantom{=}+  \left [
     \log(\frac{q}{2})-\log\left(1-\frac{q}{2}\right) + \mathcal{O}(\alpha)
   \right ]\ketbra{10} 
   \nonumber\\
   & \phantom{=}+  \left [
   \log(C(q)) - \log(\frac{q}{2}) + \mathcal{O}(\alpha \log(\alpha))
   \right ]\ketbra{11} 
   \nonumber\\
   & \phantom{=}+ \left [
   -\frac{2(q-1)}{(q-2)}\sqrt{\alpha} \log(\alpha) + \mathcal{O}(\sqrt{\alpha})
   \right ]
   \ketbra{00}{11} 
    \nonumber\\
   & \phantom{=}+ \left [
   -\frac{2(q-1)}{(q-2)}\sqrt{\alpha} \log(\alpha) + \mathcal{O}(\sqrt{\alpha})
   \right ]
   \ketbra{11}{00} \,.
   \label{Equatioin:log_minus_log}
\end{align}
Multiplying  by $\psi_{A_1B}$, we have
\begin{align}
&
\psi_{A_1 B}[ \log(\psi_{A_1 B}) - \log(\psi_{A_1}\otimes \psi_{B})] 
   \nonumber\\
   &=
   \bigg[
   \left(1-\frac{q}{2}\right)\frac{4(q-1)^2}{(q-2)^2}\alpha\log(\alpha)
   -2(1-q)\frac{2(q-1)}{(q-2)}\alpha\log(\alpha) 
   + \mathcal{O}(\sqrt{\alpha})
    \bigg]\ketbra{00} 
   +  
   \mathcal{O}(\alpha )
   \left ( \ketbra{01} + \ketbra{10} + \ketbra{11} \right )
   \nonumber\\
   &\phantom{=} + \left [
   \mathcal{O}(\sqrt{\alpha} \log(\alpha))
   \right ] \ketbra{00}{11} 
   + \left [
   \mathcal{O}(\sqrt{\alpha} \log(\alpha))
   \right ] \ketbra{11}{00}  \,.
\end{align}
Applying the trace, we approximate the relative entropy: 
\begin{align}
&
{D(\psi_{A_1B}||\psi_{A_1}\otimes \psi_{B})}
    \nonumber\\ %
    &=
    \trace\left[\psi_{A_1B}\log(\psi_{A_1B})
    -
\psi_{A_1B}\log(\psi_{A_1}\otimes \psi_{B})\right]  
    \nonumber\\
    &=
    \left [\left(1-\frac{q}{2}\right)\frac{4(q-1)^2}{(q-2)^2}
    -
    (1-q)\frac{2(q-1)}{(q-2)}\right ]\alpha\log(\alpha)
    +  \mathcal{O}(\alpha) 
    \nonumber\\
    &=
    -2\frac{(1-q)^2}{2-q}\alpha\log(\alpha) + \mathcal{O}(\alpha) \,.
    \label{Equation:Divergence}
\end{align}
 \subsection{Second Moment}
 Next, we consider the second moment of the relative entropy.
 By squaring (\ref{Equatioin:log_minus_log}), we have:
 \begin{align}
     &%
     |\log(\psi_{A_1 B}) - \log(\psi_{A_1}\otimes \psi_{B})|^2
     \nonumber\\
     &=
     \left [
    \frac{4(q-1)^2}{(q-2)^2}\alpha \log^2(\alpha) + \mathcal{O}(\alpha \log(\alpha))
    \right ]\ketbra{00} 
   \nonumber\\
   &\phantom{=} 
   +  \left [
   \mathcal{O}(\alpha^2 )
   \right ]\ketbra{01} 
   \nonumber\\
   &\phantom{=} 
   +  \left [
     \left(\log(\frac{q}{2})-\log\left(1-\frac{q}{2}\right)\right)^2 + \mathcal{O}(\alpha)
   \right ]\ketbra{10} 
   \nonumber\\
   &\phantom{=} 
   +  \Big [
  \left( \log(C(q)) - \log(\frac{q}{2}) \right)^2 
  +\frac{4(q-1)^2}{(q-2)^2}\alpha \log^2(\alpha) + \mathcal{O}(\alpha \log(\alpha))
   \Big ]\ketbra{11} 
   \nonumber\\
   &\phantom{=} 
   + \left [
    \mathcal{O}(\sqrt{\alpha}\log(\alpha))
   \right ]( \ketbra{00}{11} + \ketbra{11}{00} ) \,.
   \label{Equation:log_minus_log_square}
\end{align}
As we multiply by $\psi_{A_1 B}$, 
\begin{align}
    &\psi_{A_1 B} |\log(\psi_{A_1 B}) -
    \log(\psi_{A_1 }\otimes \psi_{B})|^2
    %
    \nonumber\\
    &= 
    \left [
    \left(1-\frac{q}{2}\right)\frac{4(q-1)^2}{(q-2)^2}\alpha \log^2(\alpha)
    + \mathcal{O}(\alpha \log(\alpha))
    \right ]\ketbra{00} 
   \nonumber\\
   &\phantom{=} 
   +  \left [
   \mathcal{O}(\alpha^2 )
   \right ]\ketbra{01} 
   +  \left [
     \mathcal{O}(\alpha)
   \right ]\ketbra{10} 
   \nonumber\\
   &\phantom{=} 
   +  \left [
   \mathcal{O}(\alpha \log(\alpha))
   \right ]\ketbra{11} 
   \nonumber\\
   &\phantom{=} 
   + \left [
    \mathcal{O}(\sqrt{\alpha}\log(\alpha)
   \right ]
   ( \ketbra{00}{11} + \ketbra{11}{00} ) \,.
\end{align}
Using (\ref{Equation:Divergence}) and applying the trace to the above, we obtain an approximation of the second moment:
\begin{align}
    V(\psi_{A_1B}||\psi_{A_1}\otimes \psi_{B}) 
    %
    & =
    \trace\Big[\psi_{A_1B}\left|\log(\psi_{A_1B})-\log(\psi_{A_1}\otimes \psi_{B})
    %
    - D(\psi_{A_1B}||\psi_{A_1}\otimes \psi_{B})\right|^2 \Big]
    \nonumber\\ &
    = \trace\left[\psi_{A_1B}\left|\log(\psi_{A_1 B})-\log(\psi_{A_1}\otimes \psi_{B})\right|^2 \right]
    \nonumber\\ &
    \phantom{=}
    - 2D(\psi_{A_1 B}||\psi_{A_1}\otimes \psi_{B})
    \trace\left[\psi_{A_1 B}\left|\log(\psi_{A_1 B})-\log(\psi_{A_1}\otimes \psi_{B})\right|\right]
    \nonumber\\ &
    \phantom{=}
    + D(\psi_{A_1 B}||\psi_{A_1}\otimes \psi_{B})^2
    \nonumber\\ &
    = \trace\left[\psi_{A_1 B}\left|\log(\psi_{A_1 B})-\log(\psi_{A_1}\otimes \psi_{B})\right|^2 \right]
    - D(\psi_{A_1 B}||\psi_{A_1}\otimes \psi_{B})^2
    \nonumber\\ &
    = \left(1-\frac{q}{2}\right)\frac{4(q-1)^2}{(q-2)^2}\alpha \log^2(\alpha)
    + \mathcal{O}(\alpha \log(\alpha)) 
    + \mathcal{O}(\alpha^2 \log^2(\alpha))
    \nonumber\\ &
    = \frac{2(q-1)^2}{q-2}\alpha \log^2(\alpha) + \mathcal{O}(\alpha \log(\alpha)) \,.
\end{align}

 \subsection{Fourth Moment}
 Consider %
\begin{align}
    {Q(\psi_{A_1 B}||\psi_{A_1}\otimes \psi_{B})}
    %
   & 
   =\trace\Big[\psi_{A_1 B}\Big|\log(\psi_{A_1 B})
    -\log(\psi_{A_1}\otimes \psi_{B})
    - D(\psi_{A_1 B}||\psi_{A_1}\otimes \psi_{B})\Big|^4 \Big] \,.
\end{align}
We use the binomial identity:
 $%
 (X-c)^4 = X^4 - 4c X^3 + 6c^2 X^2 -4c^3 X + c^4 %
$, %
for a Hermitian operator $X\in\mathcal{L}(\mathcal{H})$ and a real number $c$.
Substituting  $X=\log(\psi_{A_1 B})-\log(\psi_{A_1}\otimes \psi_{B})$, and $c=D(\psi_{A_1 B}||\psi_{A_1}\otimes \psi_{B})$, we obtain
\begin{align}
    {Q(\psi_{A_1 B}||\psi_{A_1}\otimes \psi_{B})}
    %
    &=
    \trace\left[\psi_{A_1 B}\left(\log(\psi_{A_1 B})
  -\log(\psi_{A_1}\otimes \psi_{B})\right)^4\right]
    \nonumber\\
    &\phantom{=} 
    - 4D(\psi_{A_1 B}||\psi_{A_1}\otimes \psi_{B})
    \trace\left[\psi_{A_1 B}\big(\log(\psi_{A_1 B})
    -\log(\psi_{A_1}\otimes \psi_{B})\big)^3\right]
    \nonumber\\  &\phantom{=} 
    + \mathcal{O}(\alpha^3 \log^4(\alpha)) \,.
\end{align}

Using (\ref{Equatioin:log_minus_log}) and (\ref{Equation:log_minus_log_square}), 
\begin{align}
        &%
        {(\log(\psi_{A_1 B})}-
        {\log(\psi_{A_1}\otimes \psi_{B}))^4}
    \nonumber\\
    &= \left [
    \mathcal{O}(\alpha \log^2(\alpha))
    \right ]\ketbra{00} 
   +  \left [
   \mathcal{O}(\alpha^4 )
   \right ]\ketbra{01} 
   \nonumber\\%
   &\phantom{=}+  \left [
     \mathcal{O}(1)
   \right ]\ketbra{10} 
   +  [\mathcal{O}(1)]\ketbra{11} 
   \nonumber\\
   &\phantom{=}
   + \left [
    \mathcal{O}(\sqrt{\alpha}\log(\alpha)
   \right ] \ketbra{00}{11} 
   + \left [
    \mathcal{O}(\sqrt{\alpha}\log(\alpha)
   \right ] \ketbra{11}{00} ,
   \end{align}
and
\begin{align}
        &%
        {(\log(\psi_{A_1 B})}-
        {\log(\psi_{A_1}\otimes \psi_{B}))^3} 
    \nonumber\\
    &= \left [
    \mathcal{O}(\alpha \log^2(\alpha))
    \right ]\ketbra{00} 
   +  \left [
   \mathcal{O}(\alpha^3 )
   \right ]\ketbra{01} 
    \nonumber\\
   &\phantom{=}
   +   [
     \mathcal{O}(1)
    ]\ketbra{10} 
   %
   +  
   \left[
   \mathcal{O}(1)
   \right]\ketbra{11} 
   \nonumber\\
   &\phantom{=} 
   + \left [
    \mathcal{O}(\sqrt{\alpha}\log(\alpha)
   \right ] \ketbra{00}{11} 
   + \left [
    \mathcal{O}(\sqrt{\alpha}\log(\alpha)
   \right ] \ketbra{11}{00} \,.
\end{align}

Multiplying by $\psi_{A_1B}$, we have
\begin{align}
        &\psi_{A_1 B}
        \left(\log(\psi_{A_1 B})
        -
        \log(\psi_{A_1}\otimes \psi_{B})\right)^4
    \nonumber\\
    & =\left [
    \mathcal{O}(\alpha \log^2(\alpha))
    \right ]\ketbra{00} 
   +  \left [
   \mathcal{O}(\alpha^4 )
   \right ]\ketbra{01} 
   \nonumber\\
   &\phantom{=} 
   +  \left [
     \mathcal{O}(\alpha)
   \right ] \left (\ketbra{10} +\ketbra{11} \right )
   + \left [
    \mathcal{O}(\sqrt{\alpha}\log(\alpha)
   \right ] \ketbra{00}{11} 
      \nonumber\\
   &\phantom{=} 
   + 
   \left [
    \mathcal{O}(\sqrt{\alpha}\log(\alpha)
   \right ] \ketbra{11}{00} \,,
   \end{align}
   and
\begin{align}
        &
        \psi_{A_1 B}\left|\left(\log(\psi_{A_1 B})
        \log(\psi_{A_1}\otimes \psi_{B})\right)^3\right|
    \nonumber\\
    &=
    \left [
    \mathcal{O}(\alpha \log^2(\alpha))
    \right ]\ketbra{00} 
   +  \left [
   \mathcal{O}(\alpha^3 )
   \right ]\ketbra{01} 
   \nonumber\\
   &\phantom{=} 
   +  \left [
     \mathcal{O}(\alpha)
   \right ] \left (\ketbra{10} +\ketbra{11} \right )
   \nonumber\\
   &\phantom{=} 
   + \left [
    \mathcal{O}(\sqrt{\alpha}\log(\alpha)
   \right ] (\ketbra{00}{11} + \ketbra{11}{00}) 
   \nonumber\\
   &\phantom{=} 
   + \left [
    \mathcal{O}(\sqrt{\alpha}\log(\alpha)
   \right ] \ketbra{11}{00} \,.
\end{align}
Finally, by tracing out, we obtain the order of the fourth moment:
\begin{align}
     Q(\psi_{A_1 B}&||\psi_{A_1}\otimes \psi_{B}) = \mathcal{O}(\alpha\log^2(\alpha)) \,.
\end{align}

\section{Code size}
\label{appendix:Rate_calc}
We observe that when choosing $\alpha=\alpha_n=\frac{\gamma_n}{\sqrt{n}}$ with $\gamma_n=n^{\nu-\frac{1}{6}}$, we obtain the following approximation for the first term on the right-hand side of (\ref{Equation:num_of_data_bits1}), %
\begin{align}
    & 
    n D(\psi_{A_1 B}||\psi_{A_1}\otimes \psi_{B})
    \nonumber\\
    &= -2\frac{(1-q)^2}{2-q} \sqrt{n}\gamma_n \log(\frac{\gamma_n}{\sqrt{n}})
    +O(\sqrt{n}\gamma_n) %
    \nonumber\\
    &= -2\frac{(1-q)^2}{2-q} n^{\nu+\frac{1}{3}} \log(n^{\nu-\frac{2}{3}})
    + O(n^{\nu+\frac{1}{3}}) %
    \nonumber\\
    &= 2\left (\frac{2}{3} - \nu \right ) \frac{(1-q)^2}{2-q} n^{\nu+\frac{1}{3}} \log{n}
    +O(n^{\nu+\frac{1}{3}}) \,. %
\end{align}

In a similar manner, we approximate the second term by
\begin{align}
    \sqrt{n V(\psi_{A_1 B}||\psi_{A_1}\otimes \psi_{B})}&= O\left (\sqrt{\sqrt{n}\gamma_n \log^2(n^{\nu-\frac{2}{3}})} \right )
    \nonumber\\
    &= O\left (\sqrt{n^{\nu+\frac{1}{3}}} \log{n} \right ) 
    \nonumber\\
    &=O\left (n^{\frac{\nu}{2}+\frac{1}{6}} \log{n} \right )  \,.
\end{align}
Finally,  the last term  (\ref{Equation:num_of_data_bits1}) is $C_n$, as defined in 
(\ref{Equation:C_n}). To show that this term vanishes, we write %
\begin{align}
    \frac{[Q(\psi_{A_1 B}||\psi_{A_1}\otimes \psi_{B})]^\frac{3}{4}}{V(\psi_{A_1 B}||\psi_{A_1}\otimes \psi_{B})} &=
    O\left (\frac{\left (n^{\frac{\nu}{2}+\frac{1}{6}} \log{n}\right )^\frac{3}{4}}{n^{\frac{\nu}{2}+\frac{1}{6}} \log{n}} \right )
    \nonumber\\
    &= O \left( n^{-\frac{\nu}{8}-\frac{1}{6}} \log^{-\frac{1}{4}}(n) \right)
\end{align}
which tends to zero as $n\to\infty$.
Hence,
\begin{align}
     \log(M) \geq 2\left (\frac{2}{3} - \nu \right ) \frac{(1-q)^2}{2-q} n^{\nu+\frac{1}{3}} \log{n} + O(n^{\nu+\frac{1}{3}})
\end{align}
for every $0<\nu<\frac{1}{6}$.

\section{Energy-Constrained Capacities}
\label{Appendix:Energy_Constraint}
We provide the proof for the energy-constrained capacity formula of the qubit depolarizing channel.
Note that this model does not involve a covertness requirement.

\subsection{Unassisted Capacity}
\label{Appendix:Energy_Constraint_0}
We begin with communication without assistance.
\begin{theorem}
\label{Theorem:Energy_Constraint_0}
  Consider a qubit depolarizing  channel $\mathcal{N}_{A \to B}$
    as specified in Section~\ref{Subsection:channel_model}, and let $E\in [0,1]$. 
The energy-constrained capacity without entanglement assistance is given by 
\begin{align}
C_0(\mathcal{N},E)=
\begin{cases}
h_2\left(E*\frac{q}{2}\right)-h_2\left(\frac{q}{2}\right)
& 0< E< \frac{1}{2} \\
1-h_2\left(\frac{q}{2}\right)
&  \frac{1}{2} \leq  E\leq 1,
\end{cases}
\label{Equation:Energy_Capacity_0_App}
\end{align}
where `$*$' denotes the binary convolution operation: $\alpha*\beta=(1-
\alpha)\beta+\alpha(1-\beta)$. 
\end{theorem}

\begin{proof}
Consider the general capacity characterization in \eqref{Eqution:Unassisted_Capacity}.
For $0< E< \frac{1}{2}$,
the direct part follows by choosing the ensemble
$\left\{ \left( 1-E,E \right), \ket{0},\ket{1}  \right\}$, which results in the average input state 
\begin{align}
\widetilde{\rho}_A \equiv (1-E)\ketbra{0} + E\ketbra{1} \,.
\end{align}
Otherwise, if $E\geq \frac{1}{2}$, set the input ensemble to be uniform, i.e., $\left\{ \left( \frac{1}{2},\frac{1}{2} \right), \ket{0},\ket{1}  \right\}$

We move to the converse part.
For $E\geq \frac{1}{2}$, the converse part immediately follows from the capacity result without constraints. Hence, suppose that $0< E< \frac{1}{2}$.
For every input ensemble, the Holevo information functional, $I(X;B)_\rho$, is bounded as follows:
\begin{align}
    I(X;B)_\rho &= H(B)_{\rho} - H(B|X)_{\rho}
    \nonumber\\
    & \leq H(B)_{\rho}-H^{\min}(\mathcal{N})
\end{align}
where $H^{\min}(\mathcal{N})$ is the minimum output entropy,
$
H^{\min}(\mathcal{N})\equiv \min_{\rho_A} H(\mathcal{N}(\rho_A)) 
$. 
For the qubit depolarizing channel,
\begin{align}
H^{\min}(\mathcal{N})=h_2\left( \frac{q}{2} \right)
\end{align}
by \cite[Sec. 20.4.4]{wilde2013quantum}. 

It remains to bound $H(B)_{\rho}$.
Consider a general input state 
\begin{align}
\rho_A \equiv 
\begin{pmatrix}
1-a & b \\ b^* & a 
\end{pmatrix} 
\end{align}
that satisfies the maximization constraint, 
$\trace(\mathsf{F} \rho_A)\leq E$
(see \eqref{Eqution:Unassisted_Capacity}). 
%
%
%
The corresponding output state is  
\begin{align}
\rho_B \equiv 
\begin{pmatrix}
(1-q)(1-a) + \frac{q}{2} & (1-q)b \\ (1-q)b^* & (1-q)a + \frac{q}{2}
\end{pmatrix}\,.
\end{align}
%
The eigenvalues of $\rho_B$ are thus
\begin{align}
    &\pi_{\pm} = \frac{1}{2}\left(1
    \,
    \pm 
    \sqrt{1 - 4 \left (((1-q)(1-a) + \frac{q}{2})((1-q)a + \frac{q}{2})\right ) + 4 |b|^2} \right) \,.
\end{align}
Hence, the output entropy  is
\begin{align}
    H(\rho_B) = -\pi_+ \log(\pi_+) - \pi_- \log(\pi_-) \,.
\end{align}
Notice that  the eigenvalues $\pi_\pm$ do not depend on the phase of the off-diagonal entry, $b$, hence the entropies of $\rho_A$ and $Z\rho_A Z$ are the same. It thus follows that
\begin{align}
    H(\rho_B) = H(Z \rho_B Z)
\end{align}
with
\begin{align}
Z\rho_B Z \equiv 
\begin{pmatrix}
(1-q)(1-a) + \frac{q}{2} & -(1-q)b \\ -(1-q)b^* & (1-q)a + \frac{q}{2}
\end{pmatrix}\,.
\end{align}
Since the entropy is concave, we have
\begin{align}
    &H(\rho_B) =\frac{1}{2}H(\rho_B) + \frac{1}{2}H(Z\rho_B Z)
    \nonumber\\
    &\leq H\left(\frac{1}{2}\rho_B + \frac{1}{2}Z\rho_B Z\right) 
    \nonumber\\
    &= H \left ( \left[(1-q)(1-a) + \frac{q}{2} \right] \ketbra{0}
    + \left[(1-q)a + \frac{q}{2}\right] \ketbra{1} \right )
    \nonumber\\
    &= H\Big( \mathcal{N}_{A \to B}\left ( (1-a)\ketbra{0} + a \ketbra{1} \right ) \Big ).
\end{align}
Therefore, the maximal output entropy can be achieved with $b=0$. i.e., for an input  state of the form
\begin{align}
\rho_A \equiv 
\begin{pmatrix}
1-a & 0 \\ 0 & a 
\end{pmatrix} \,.
\end{align}
Since the energy constraint requires $a\leq E$,
\begin{align}
H(B)_\rho 
&\leq \max_{0\leq a\leq E} H\Big( \mathcal{N}_{A \to B}\left ( (1-a)\ketbra{0} + a \ketbra{1} \right ) \Big )
\nonumber\\
&= \max_{0\leq a\leq E} h_2\Big( a*\frac{q}{2} \Big )
\nonumber\\
&=  h_2\Big( E*\frac{q}{2} \Big )
\end{align}
This completes the proof of Theorem~\ref{Theorem:Energy_Constraint_0}.
\end{proof}

\subsection{Entanglement-Assisted Capacity}
We move to the energy-constrained capacity of the qubit depolarizing channel, when Alice and Bob are provided with pre-shared entanglement. 

\begin{theorem}
\label{Theorem:Energy_Constraint_EA}
The energy-constrained entanglement-assisted capacity  is given by 
\begin{align}
&C_{\text{EA}}(\mathcal{N},E)=
\begin{cases}
h_2(E)+h_2\left(E*\frac{q}{2}\right)-H(\Psi_{A_1 B})
& 0< E< \frac{1}{2} \\
2-H\left(1-\frac{3q}{4},\frac{q}{4},\frac{q}{4},\frac{q}{4}\right)
&  \frac{1}{2} \leq  E\leq 1,
\end{cases}
\label{Equation:Energy_Capacity_EA_App}
\end{align}
 where $\Psi_{A_1 B}\equiv (\mathrm{id}\otimes\mathcal{N})(\ketbra{\Psi_{A_1 A}})$. 
\end{theorem}
\begin{proof}
Recall that the entanglement-assisted capacity of a general channel $\mathcal{N}_{A\to B}$ is given by \ref{Equation:EA_Capacity}. For the qubit depolarizing channel,
we can restrict our attention to input states of the form 
$\ket{\psi_{A_1 A}} = \sqrt{1-a}\ket{00} + \sqrt{a}\ket{11}$, since the depolarizing channel is unitarily covariant (see 
\cite[Section~24.8]{wilde2013quantum}).
For $E<\frac{1}{2}$, the maximum is attained by the entangled state
\begin{align}
\ket{\Psi_{A_1 A}} = \sqrt{1-E}\ket{00} + \sqrt{E}\ket{11}
\,,
\end{align}
which is associated with an  energy value $\trace(\mathsf{F}\psi_A) = E$.
Whereas, for $E\geq \frac{1}{2}$ the capacity is attained with $a=\frac{1}{2}$.
This completes the proof of Theorem~\ref{Theorem:Energy_Constraint_EA}.
\end{proof} 
 \end{appendices}

 {
\bibliography{bibliography}

\begin{thebibliography}{10}
\providecommand{\url}[1]{#1}
\csname url@samestyle\endcsname
\providecommand{\newblock}{\relax}
\providecommand{\bibinfo}[2]{#2}
\providecommand{\BIBentrySTDinterwordspacing}{\spaceskip=0pt\relax}
\providecommand{\BIBentryALTinterwordstretchfactor}{4}
\providecommand{\BIBentryALTinterwordspacing}{\spaceskip=\fontdimen2\font plus
\BIBentryALTinterwordstretchfactor\fontdimen3\font minus
  \fontdimen4\font\relax}
\providecommand{\BIBforeignlanguage}[2]{{%
\expandafter\ifx\csname l@#1\endcsname\relax
\typeout{** WARNING: IEEEtran.bst: No hyphenation pattern has been}%
\typeout{** loaded for the language `#1'. Using the pattern for}%
\typeout{** the default language instead.}%
\else
\language=\csname l@#1\endcsname
\fi
#2}}
\providecommand{\BIBdecl}{\relax}
\BIBdecl

\bibitem{wang2020security}
M.~Wang, T.~Zhu, T.~Zhang, J.~Zhang, S.~Yu, and W.~Zhou, ``Security and privacy
  in {6G} networks: New areas and new challenges,'' \emph{Digital Commun.
  Netw.}, vol.~6, no.~3, pp. 281--291, 2020.

\bibitem{talbot2006complexity}
J.~Talbot, D.~Welsh, and D.~J.~A. Welsh, \emph{Complexity and cryptography: An
  Introduction}.\hskip 1em plus 0.5em minus 0.4em\relax Cambridge University
  Press, 2006, vol.~13.

\bibitem{bloch2011physical}
M.~Bloch and J.~Barros, \emph{Physical-layer Security: From Information Theory
  to Security engineering}.\hskip 1em plus 0.5em minus 0.4em\relax Cambridge
  University Press, 2011.

\bibitem{bennet1984quantum}
C.~H. Bennett and G.~Brassard, ``Quantum cryptography: Public key distribution
  and coin tossing,'' in \emph{Proc. IEEE Int. Conf. Comp.}, 1984.

\bibitem{renner2008security}
R.~Renner, ``Security of quantum key distribution,'' \emph{Int. J. Quantum
  Inf.}, vol.~6, no.~01, pp. 1--127, 2008.

\bibitem{scarani2009security}
V.~Scarani, H.~Bechmann-Pasquinucci, N.~J. Cerf, M.~Du{\v{s}}ek,
  N.~L{\"u}tkenhaus, and M.~Peev, ``The security of practical quantum key
  distribution,'' \emph{Rev. Mod. Phys.}, vol.~81, no.~3, p. 1301, 2009.

\bibitem{9380147}
M.~Bloch, O.~G{\"u}nl{\"u}, A.~Yener, F.~Oggier, H.~V. Poor, L.~Sankar, and
  R.~F. Schaefer, ``An overview of information-theoretic security and privacy:
  Metrics, limits and applications,'' \emph{IEEE J. Sel. Areas Inf. Theory},
  vol.~2, no.~1, pp. 5--22, 2021.

\bibitem{7217803}
R.~F. Schaefer, H.~Boche, and H.~V. Poor, ``Secure communication under channel
  uncertainty and adversarial attacks,'' \emph{Proc. IEEE}, vol. 103, no.~10,
  pp. 1796--1813, 2015.

\bibitem{bash2015hiding}
B.~A. Bash, D.~Goeckel, D.~Towsley, and S.~Guha, ``Hiding information in noise:
  Fundamental limits of covert wireless communication,'' \emph{IEEE Commun.
  Mag.}, vol.~53, no.~12, pp. 26--31, 2015.

\bibitem{8917582}
M.~Tahmasbi, A.~Savard, and M.~R. Bloch, ``Covert capacity of non-coherent
  rayleigh-fading channels,'' \emph{IEEE Trans. Inf. Theory}, vol.~66, no.~4,
  pp. 1979--2005, 2020.

\bibitem{Bash_first_classic}
B.~A. Bash, D.~Goeckel, and D.~Towsley, ``Limits of reliable communication with
  low probability of detection on {AWGN} channels,'' \emph{IEEE J. Sel. Areas
  Commun.}, vol.~31, no.~9, pp. 1921--1930, 2013.

\bibitem{Bloch}
M.~R. Bloch, ``Covert communication over noisy channels: A resolvability
  perspective,'' \emph{IEEE Trans. Inf. Theory}, vol.~62, no.~5, pp.
  2334--2354, 2016.

\bibitem{wang2016fundamental}
L.~Wang, G.~W. Wornell, and L.~Zheng, ``Fundamental limits of communication
  with low probability of detection,'' \emph{IEEE Trans. Inf. Theory}, vol.~62,
  no.~6, pp. 3493--3503, 2016.

\bibitem{Bash-Quantum}
M.~S. Bullock, A.~Sheikholeslami, M.~Tahmasbi, R.~C. Macdonald, S.~Guha, and
  B.~A. Bash, ``Fundamental limits of covert communication over
  classical-quantum channels,'' \emph{IEEE Trans. Inf. Theory}, 2025.

\bibitem{wang16cq-srlconverse}
L.~Wang, ``Optimal throughput for covert communication over a classical-quantum
  channel,'' in \emph{Proc. Inform. Theory Workshop ({ITW})}, Cambridge, UK,
  Sep. 2016, pp. 364--368, arXiv:1603.05823 [cs.IT].

\bibitem{Bash-Quantum-revision}
M.~S. Bullock, A.~Sheikholeslami, M.~Tahmasbi, R.~C. Macdonald, S.~Guha, and
  B.~A. Bash, ``Covert communication over classical-quantum channels,''
  arXiv:1601.06826 [quant-ph], 2023.

\bibitem{tahmasbi2019framework}
M.~Tahmasbi and M.~R. Bloch, ``Framework for covert and secret key expansion
  over classical-quantum channels,'' \emph{Phys. Rev. A}, vol.~99, no.~5, p.
  052329, 2019.

\bibitem{bash_first_bosonic}
B.~A. Bash, A.~H. Gheorghe, M.~Patel, J.~L. Habif, D.~Goeckel, D.~Towsley, and
  S.~Guha, ``Quantum-secure covert communication on bosonic channels,''
  \emph{Nat. commun.}, vol.~6, no.~1, pp. 1--9, 2015.

\bibitem{8976410}
M.~S. Bullock, C.~N. Gagatsos, S.~Guha, and B.~A. Bash, ``Fundamental limits of
  quantum-secure covert communication over bosonic channels,'' \emph{IEEE J.
  Sel. Areas Commun.}, vol.~38, no.~3, pp. 471--482, 2020.

\bibitem{Bash-Bosonic}
C.~N. Gagatsos, M.~S. Bullock, and B.~A. Bash, ``Covert capacity of bosonic
  channels,'' \emph{IEEE J. Sel. Areas Inf. Theory}, vol.~1, no.~2, pp.
  555--567, 8 2020.

\bibitem{9834394}
S.-Y. Wang, T.~Erdoğan, and M.~Bloch, ``Towards a characterization of the
  covert capacity of bosonic channels under trace distance,'' in \emph{Proc.
  IEEE Int. Symp. Inf. Theory (ISIT)}, 2022, pp. 318--323.

\bibitem{bash2017fundamental_sensing}
B.~A. Bash, C.~N. Gagatsos, A.~Datta, and S.~Guha, ``Fundamental limits of
  quantum-secure covert optical sensing,'' in \emph{Proc. IEEE Int. Symp. Inf.
  Theory (ISIT)}.\hskip 1em plus 0.5em minus 0.4em\relax IEEE, 2017, pp.
  3210--3214.

\bibitem{goeckel17sensinglinsystems-asilomar}
D.~Goeckel, B.~A. Bash, A.~Sheikholeslami, S.~Guha, and D.~Towsley, ``Covert
  active sensing of linear systems,'' in \emph{Proc. Asilomar Conf. Signals
  Syst. Comput.}, Pacific Grove, CA, USA, Nov. 2017.

\bibitem{tahmasbi2021covert}
M.~Tahmasbi and M.~R. Bloch, ``On covert quantum sensing and the benefits of
  entanglement,'' \emph{IEEE J. Sel. Areas Inf. Theory}, vol.~2, no.~1, pp.
  352--365, 2021.

\bibitem{Deng2022}
D.~Deng, X.~Li, S.~Dang, M.~C. Gursoy, and A.~Nallanathan, ``Covert
  communications in intelligent reflecting surface-assisted two-way relaying
  networks,'' \emph{IEEE Transactions on Vehicular Technology}, vol.~71,
  no.~11, pp. 12\,380--12\,385, 2022.

\bibitem{ZivariFard2022}
H.~ZivariFard, M.~R. Bloch, and A.~Nosratinia, ``Covert communication in the
  presence of an uninformed, informed, and coordinated jammer,'' in \emph{Proc.
  IEEE Int. Symp. Inf. Theory (ISIT)}, 2022, pp. 306--311.

\bibitem{Yang2022}
B.~Yang, T.~Taleb, G.~Chen, and S.~Shen, ``Covert communication for cellular
  and x2u-enabled uav networks with active and passive wardens,'' \emph{IEEE
  Network}, vol.~36, no.~1, pp. 166--173, 2022.

\bibitem{Amihood:ITW2022}
B.~Amihood and A.~Cohen, ``Covertly controlling a linear system,'' in
  \emph{2022 IEEE Information Theory Workshop (ITW)}, 2022, pp. 321--326.

\bibitem{hayashi2023covert}
M.~Hayashi and A.~Vazquez-Castro, ``Covert communication with gaussian noise:
  from random access channel to point-to-point channel,'' \emph{arXiv preprint
  arXiv:2310.15519}, 2023.

\bibitem{bounhar2023mixing}
A.~Bounhar, M.~Sarkiss, and M.~Wigger, ``Mixing a covert and a non-covert
  user,'' \emph{arXiv preprint arXiv:2305.06268}, 2023.

\bibitem{bennett1999entanglement}
C.~H. Bennett, P.~W. Shor, J.~A. Smolin, and A.~V. Thapliyal,
  ``Entanglement-assisted classical capacity of noisy quantum channels,''
  \emph{Phys. Rev. Lett}, vol.~83, no.~15, p. 3081, 1999.

\bibitem{EA_capacity_bennet}
------, ``Entanglement-assisted capacity of a quantum channel and the reverse
  shannon theorem,'' \emph{IEEE Trans. Inf. Theory}, vol.~48, no.~10, pp.
  2637--2655, 2002.

\bibitem{hao2021entanglement}
S.~Hao, H.~Shi, W.~Li, J.~H. Shapiro, Q.~Zhuang, and Z.~Zhang,
  ``Entanglement-assisted communication surpassing the ultimate classical
  capacity,'' \emph{Phys. Rev. Lett.}, vol. 126, no.~25, p. 250501, 2021.

\bibitem{chiuri2013experimental}
A.~Chiuri, S.~Giacomini, C.~Macchiavello, and P.~Mataloni, ``Experimental
  achievement of the entanglement-assisted capacity for the depolarizing
  channel,'' \emph{Phys. Rev. A}, vol.~87, no.~2, p. 022333, 2013.

\bibitem{9319007}
U.~Pereg, C.~Deppe, and H.~Boche, ``Quantum channel state masking,'' \emph{IEEE
  Trans. Inf. Theory}, vol.~67, no.~4, pp. 2245--2268, 2021.

\bibitem{bloch_resource_efficient}
S.-Y. Wang, S.-J. Su, and M.~Bloch, ``Resource-efficient entanglement-assisted
  covert communications over bosonic channels,'' \emph{Proc. IEEE Int. Symp.
  Inf. Theory (ISIT)}, 2024.

\bibitem{ChiuriGiacominiMacchiavelloMataloni:13p}
A.~Chiuri, S.~Giacomini, C.~Macchiavello, and P.~Mataloni, ``Experimental
  achievement of the entanglement-assisted capacity for the depolarizing
  channel,'' \emph{Phys. Rev. A}, vol.~87, no.~2, p. 022333, 2013.

\bibitem{Google:19p}
F.~Arute, K.~Arya, R.~Babbush, D.~Bacon, J.~C. Bardin, R.~Barends, R.~Biswas,
  S.~Boixo, F.~G. S.~L. Brandao, D.~A. Buell \emph{et~al.}, ``Quantum supremacy
  using a programmable superconducting processor,'' \emph{Nature}, vol. 574,
  no. 7779, pp. 505--510, 2019.

\bibitem{king2003capacity}
C.~King, ``The capacity of the quantum depolarizing channel,'' \emph{IEEE
  Trans. Inf. Theory}, vol.~49, no.~1, pp. 221--229, 2003.

\bibitem{depolarizing_comp}
D.~Leung and J.~Watrous, ``On the complementary quantum capacity of the
  depolarizing channel,'' \emph{Quantum}, vol.~1, 10 2015.

\bibitem{wilde2013quantum}
M.~M. Wilde, \emph{Quantum information theory}, 2nd~ed.\hskip 1em plus 0.5em
  minus 0.4em\relax Cambridge University Press, 2017.

\bibitem{9834764}
U.~Pereg, C.~Deppe, and H.~Boche, ``Communication with unreliable entanglement
  assistance,'' \emph{Preprint available on \texttt{arXiv:2112.09227}. Proc.
  IEEE Int. Symp. Inf. Theory (ISIT)}, pp. 2231--2236, 2022.

\bibitem{tomamichel2015quantum}
M.~Tomamichel, \emph{Quantum information processing with finite resources:
  mathematical foundations}.\hskip 1em plus 0.5em minus 0.4em\relax Springer,
  2015, vol.~5.

\bibitem{cormen2022introduction}
T.~H. Cormen, C.~E. Leiserson, R.~L. Rivest, and C.~Stein, \emph{Introduction
  to algorithms}.\hskip 1em plus 0.5em minus 0.4em\relax MIT press, 2022.

\bibitem{4475375}
D.~Kretschmann, D.~Schlingemann, and R.~F. Werner, ``The
  information-disturbance tradeoff and the continuity of stinespring's
  representation,'' \emph{IEEE Trans. Inf. Theory}, vol.~54, no.~4, pp.
  1708--1717, 2008.

\bibitem{deterministic_identification_Gaussian_Salariseddigh}
M.~J. Salariseddigh, U.~Pereg, H.~Boche, and C.~Deppe, ``Deterministic
  identification over fading channels,'' in \emph{IEEE Inf. Theory Workshop
  (ITW)}, 2021, pp. 1--5.

\bibitem{salariseddigh2022deterministic_poisson}
M.~J. Salariseddigh, U.~Pereg, H.~Boche, C.~Deppe, V.~Jamali, and R.~Schober,
  ``Deterministic identification for molecular communications over the poisson
  channel,'' \emph{arXiv:2203.02784}, 2022.

\bibitem{Ahlswede1989IdentificationVC}
R.~Ahlswede and G.~Dueck, ``Identification via channels,'' \emph{IEEE Trans.
  Inf. Theory}, vol.~35, pp. 15--29, 1989.

\bibitem{anshu2017one}
A.~Anshu, R.~Jain, and N.~A. Warsi, ``One shot entanglement assisted classical
  and quantum communication over noisy quantum channels: A hypothesis testing
  and convex split approach,'' \emph{arXiv preprint arXiv:1702.01940}, 2017.

\bibitem{khabbazi2019union}
S.~Khabbazi~Oskouei, S.~Mancini, and M.~M. Wilde, ``Union bound for quantum
  information processing,'' \emph{Proc. Royal Soc. A}, vol. 475, no. 2221, p.
  20180612, 2019.

\bibitem{wilde2017position}
M.~M. Wilde, ``Position-based coding and convex splitting for private
  communication over quantum channels,'' \emph{Quantum Inf. Proc.}, vol.~16,
  no.~10, p. 264, 2017.

\bibitem{holevo2003entanglement}
A.~S. Holevo, ``Entanglement-assisted capacity of constrained channels,'' in
  \emph{1st Int. Symp. Quantum Info.}, vol. 5128.\hskip 1em plus 0.5em minus
  0.4em\relax SPIE, 2003, pp. 62--69.

\bibitem{9173940}
S.~Guha, Q.~Zhuang, and B.~A. Bash, ``Infinite-fold enhancement in
  communications capacity using pre-shared entanglement,'' in \emph{Proc. IEEE
  Int. Symp. Inf. Theory (ISIT)}, 2020, pp. 1835--1839.

\bibitem{Shi_2020}
H.~Shi, Z.~Zhang, and Q.~Zhuang, ``Practical route to entanglement-assisted
  communication over noisy bosonic channels,'' \emph{Phys. Rev. App.}, vol.~13,
  no.~3, mar 2020.

\bibitem{bennett1997capacities}
C.~H. Bennett, D.~P. DiVincenzo, and J.~A. Smolin, ``Capacities of quantum
  erasure channels,'' \emph{Phys. Rev. Lett.}, vol.~78, no.~16, p. 3217, 1997.

\end{thebibliography}
}

\end{document}